\documentclass[letterpaper,twocolumn,10pt]{article}
\usepackage{usenix}

\usepackage{tikz}

\usepackage{amsmath,amsfonts}
\usepackage{algorithm}
\usepackage{array}
\usepackage[caption=false]{subfig}
\usepackage{textcomp}
\usepackage{verbatim}
\usepackage{graphicx}
\usepackage{url}
\usepackage{tabularx}
\usepackage{algpseudocode}
\usepackage{amssymb,amsthm}
\usepackage{booktabs,multirow,makecell}
\usepackage[normalem]{ulem}
\usepackage{threeparttable}
\usepackage{placeins}
\usepackage{etoolbox}
\usepackage{float}

\newcommand{\best}[1]{\underline{\textbf{#1}}}
\newcolumntype{L}[1]{>{\hsize=#1\hsize\linewidth=\hsize\raggedright\arraybackslash}X}
\newcolumntype{C}[1]{>{\hsize=#1\hsize\linewidth=\hsize\centering\arraybackslash}X}
\newcolumntype{R}[1]{>{\hsize=#1\hsize\linewidth=\hsize\raggedleft\arraybackslash}X}
\newcommand{\PGTableNote}[1]{\par\medskip\raggedright #1\par}

\DeclareMathOperator{\Restore}{Restore}

\newtheorem{proposition}{Proposition}

\begin{document}

\date{}

\title{\Large \bf PrivacyGuard: Source-Linked Sanitization and Locally Resolvable PII References for Cloud GUI Agents}

\author{
{\rm Zian Yuan}\\
\texttt{zian.yuan@outlook.com}
\and
{\rm Wei Li}\thanks{Corresponding author.}\\
The University of Sydney\\
\texttt{weiwilson.li@sydney.edu.au}
}

\maketitle

\begin{abstract}
Cloud graphical user interface (GUI) agents send instructions and screenshots to remote planners, exposing personally identifiable information (PII). Existing placeholder defenses mask detected spans and reuse exact strings. Fragmented or changed forms can expose PII or break entity continuity across turns. PrivacyGuard addresses this private reference problem with source-linked sanitization (SAN) and a structured reference lifecycle (REF). SAN reconstructs fragmented or overlapping evidence and rewrites the complete source extent. REF links changed appearances to session entities and supports locally resolvable references without returning raw values to the upload path. In five-turn sessions scored against common gold, PrivacyGuard reduces audited source-value leakage by 33.55\%--82.60\% relative to the closest prior full-pipeline system in all four reader--language combinations. REF raises strict cross-form continuity from 0.04 under exact-value reuse to 0.81. In reference-required Android tasks, verified completion rises from 8/30 and 7/30 to 17/30 and 20/30. On AndroidLab, PrivacyGuard achieves the highest or tied-highest task-success point estimate among protection pipelines for all four planners. On PrivScreen, it reduces fixed-reader field exposure by 62.72\%--62.84\% relative to the same baseline. PrivacyGuard lowers upload exposure while preserving a locally resolvable task path.
\end{abstract}

\section{Introduction}
\label{sec:Introduction}

Cloud-based graphical user interface (GUI) agents send an instruction and the current screen to a remote planner before device execution begins~\cite{zhang2024large}. Personally identifiable information (PII) can appear directly in an uploaded view. Reusable references create a second privacy risk: visible type, equality, projection, and layout cues may support reconstruction when several protected uploads are examined together~\cite{zhao2026webpii,agyemang2026llmredactor,jathanna2026surrogateshield}. PrivacyGuard reduces exposure at this boundary by combining source-linked sanitization (SAN) with a structured reference lifecycle (REF). We use \emph{Full} for the combined SAN+REF pipeline. SAN rewrites complete source extents before upload, and REF preserves a locally resolvable task path without returning raw values to the upload path.

The problem arises when a protected entity appears in different forms. Consider a delivery task containing \emph{350 5th Avenue, Apt.\ 24B, New York, NY}. A later screen may show \emph{5th Ave.} in one optical character recognition (OCR) region and \emph{Apt.\ 24B} in another, followed by a request for the city. These fragments refer to the same destination, but neither is an exact copy of the original value. This case exposes the need for entity-level continuity, which keeps changed appearances linked and supports later component requests without exposing the raw address. We call this the \emph{private reference problem}: the protection unit is an entity and its references, not only a repeated string. Privacy requires removing the value from each upload, but task completion may still depend on the entity represented by that value.

Existing value-placeholder systems do not address both sides of the problem. AbI~\cite{zhao2026anonymization}, the closest prior full-pipeline system, uses deterministic placeholders, a session value map with exact lookup and fuzzy OCR alignment, and whole-value local restoration. Such value-level mapping cannot represent a complete value and later abbreviations, fragments, or component-only appearances as structured views of the same entity. Whole-value restoration cannot express a request for one supported component. The private reference problem has two coupled parts. The first is recovering complete, rewriteable source extents from fragmented OCR boxes and partial or overlapping detector outputs. The second is preserving entity continuity after source removal so that a later task can refer to the protected entity without releasing its raw value.

PrivacyGuard parses each accepted occurrence once into a source trace and a typed entity record. PrivacyGuard reconstructs fragmented evidence and resolves competing candidates before finalizing the source trace, which preserves the prompt offsets or screenshot regions that produced the occurrence. SAN uses this trace to rewrite the complete source extent instead of the detected span. This source-linked design reduces the risk that partial detections leave visible characters or divide one occurrence into unrelated protected regions.

The typed entity record supports REF. It retains the entity identity and locally available components needed to connect changed appearances. Compatible appearances reuse a session root, and supported address components can be requested by projection and resolved against local session state. Addresses provide the main evaluated projection case. The remaining supported PII types retain type-appropriate state. This distinction matters when a protected value is needed for task completion. Removing such a value can reduce GUI-agent success~\cite{wang2026guiguard}, but REF preserves a locally resolvable path without placing the raw value back into the upload path.

We evaluate task execution with protocol-conforming remote planners and privacy exposure with fixed readers. A released-view audit examines one sanitized instruction or screenshot, and a five-turn audit examines the ordered Full uploads from the same session. The Full-pipeline privacy endpoint, \emph{audited source-value leakage}, is the deduplicated union of per-view exposure and five-turn reconstruction. Task success remains a separate endpoint. Controlled and external experiments isolate SAN source completeness, REF cross-form continuity and reference-mediated execution, and the resulting Full-pipeline privacy and utility.

PrivacyGuard lowers audited source-value leakage by 33.55\%--82.60\% relative to the closest prior full-pipeline system in all four reader--language combinations. REF raises strict cross-form continuity from 0.04 with exact-value reuse to 0.81. In reference-required Android tasks, verified completion rises from 8/30 and 7/30 with exact-value reuse to 17/30 and 20/30 with PrivacyGuard. On the 138-task AndroidLab benchmark, PrivacyGuard achieves the highest or tied-highest task-success point estimate among protection pipelines for all four planners. On PrivScreen, it lowers fixed-reader field exposure by 62.72\%--62.84\% relative to the same full-pipeline baseline. A Xiaomi 12S workload demonstrates device-side feasibility of the packaged prototype.

This work makes three contributions:

\begin{itemize}

    \item \textbf{Source-linked sanitization.}
    PrivacyGuard addresses the private reference problem by reconstructing fragmented or overlapping evidence and using source traces to rewrite complete source extents. This design improves complete-occurrence coverage and reduces residual source characters on controlled and external inputs.

    \item \textbf{Structured reference lifecycle.}
    Typed entity records support session-stable roots, changed appearances, address-component projections, and deterministic local restoration. REF preserves a locally resolvable task path after source removal and improves cross-form continuity and verified task completion over exact-value reuse.

    \item \textbf{Full-pipeline evaluation.}
    CrossTurnPII provides bilingual gold annotations for source extent, entity identity, and projection in five-turn sessions. PrivScreen, AndroidLab, reference-required live tasks, and a Xiaomi 12S workload evaluate external privacy, task utility, and packaged-prototype feasibility while keeping privacy and task-success outcomes separate.

\end{itemize}

\section{System Model and Goals}
\label{sec:overview}

PrivacyGuard places SAN and REF in one mediation path between a trusted GUI-agent client and a remote planner. An accepted occurrence is a PII candidate that passes adjudication and produces a valid source trace and typed entity record. For each accepted occurrence, SAN follows its source trace to rewrite the linked content in the released instruction or screenshot. REF uses the corresponding entity record to render a structured reference backed by local session state. Compatible appearances bound to one entity reuse a session-scoped root, and supported projections allow the planner to request locally available address components. Returned references are validated and resolved locally before the trusted executor receives the action payload. Action choice, sink authorization, and execution remain GUI-agent responsibilities outside PrivacyGuard. Figure~\ref{fig:pg-architecture} summarizes the upload mediation pipeline, the two local-only stores, the reference-bearing cloud exchange, and the local restoration path.

\begin{figure*}[t]
\centering
\includegraphics[width=1\textwidth]{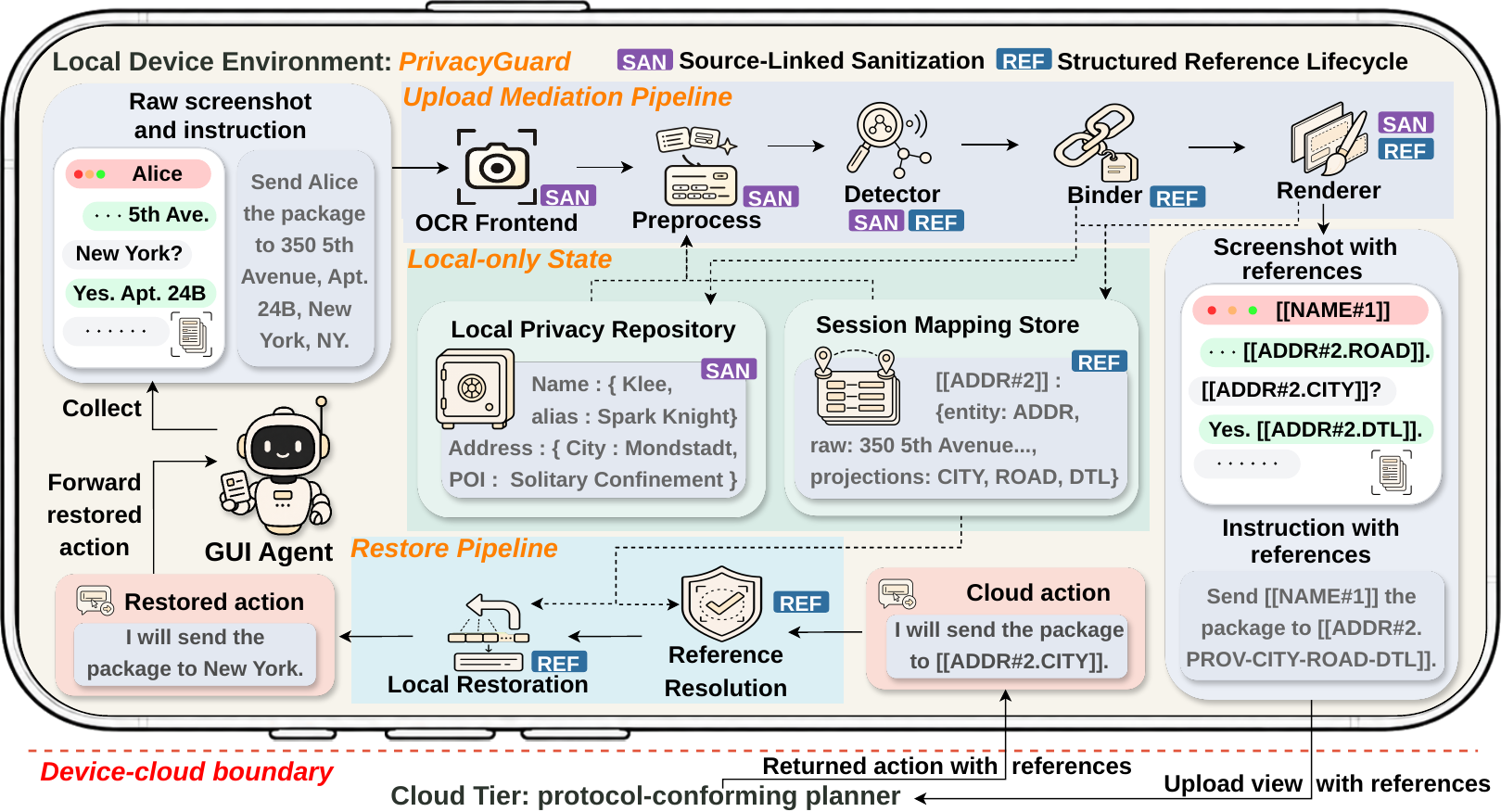}
\caption{PrivacyGuard's device-resident architecture. OCR and processing feed a shared detection and interpretation path. SAN rewrites accepted source-linked occurrences, while REF binds the resulting typed entities to session references and renders them in the released instruction or screenshot. The local privacy repository supplies detection evidence and the session mapping store retains roots, projections, and raw values on the trusted client. A protocol-conforming remote planner receives the released view and may return reference-bearing actions, which are validated, resolved, and restored locally before GUI-agent execution.}
\label{fig:pg-architecture}
\end{figure*}

\subsection{Mediation and Local State}
\label{subsec:local-state}

The client retains the raw view and two distinct stores. The optional \emph{local privacy repository} supplies user-linked values and aliases only as detection evidence. The ephemeral \emph{session mapping store} contains issued roots and their typed entity records and alone governs local resolution. References are ordinary strings. \texttt{[[TYPE\#N]]} denotes a root, and a suffix in \texttt{[[TYPE\#N.SPEC]]} denotes either a supported address component such as \texttt{CITY} or descriptive \texttt{LEN=L} metadata for NUM/ALNUM.

Compatible observations update an issued entity record only under the binding and conflict rules. Address updates merge compatible components and rebuild the canonical surface. Other types retain type-compatible normalized state. A bare root resolves to the current stored value. An address projection resolves only when the requested component is supported and locally present in the entity record. The contract is root-scoped, so a supported component need not have appeared previously as a rendered projection. Syntactically valid NUM/ALNUM \texttt{LEN=L} metadata is descriptive and does not affect authorization or value-integrity checking.

Malformed reference-like strings remain literal. Parsed requests with inactive roots, type mismatches, or unsupported or unavailable projections remain unresolved and literal. For fixed issued session state and returned text, resolution is independent of the local privacy repository. Appendix~\ref{sec:appendix-store-separation} gives the repository schema, update semantics, and formal separation property. Appendix~\ref{sec:appendix-topology} describes the deployment profiles and state lifetimes.

For example, a full destination and later \emph{5th Ave.}/\emph{Apt.\ 24B} fragments may bind to \texttt{ADDR\#2}. The planner can return \texttt{[[ADDR\#2.CITY]]}, which the client resolves to \emph{New York} when the active entity record contains that component. The raw address is removed before upload. The current implementation protects names, phone numbers, email addresses, organizations, addresses, and context-supported numeric or alphanumeric values. All supported types use parser-internal structure for detection or binding. Addresses are the main evaluated projection case.

\subsection{Trust Boundary and Adversary Model}
\label{subsec:security-scope}

In the intended deployment, the uncompromised client device, PrivacyGuard mediation path, local state stores, and the GUI agent's separate executor are trusted. For controlled comparisons, a trusted evaluation server hosts the same mediation function for every evaluated protection pipeline. This server is experimental apparatus, not the remote planner, and does not alter the device-resident deployment model. PrivacyGuard treats the remote planner as honest-but-curious for confidentiality. The planner follows the normal GUI-agent protocol but knows PrivacyGuard's design and public reference contract. It may retain and jointly analyze all released uploads and its own actions within the session using public application knowledge. The fixed readers instantiate the released-upload component of this passive reconstruction capability. They do not define the adversary model.

Planner output remains untrusted even during normal protocol execution. PrivacyGuard applies the local session contract to every returned reference-like string before resolution. The planner receives no direct resolution-success signal. A later uploaded screen may reveal effects of local execution, but that screen first passes through SAN as a new source view. Deliberate adaptive oracle probing that varies returned actions to infer local state from execution feedback is outside scope. PrivacyGuard excludes client compromise, mediation bypass, side channels, availability attacks, and action-level misuse. Execution integrity and sink authorization remain responsibilities of the GUI-agent layer. A client-supplied session namespace corresponds to a continuing goal and is reset or expires at goal boundaries. The confidentiality claim is session-scoped and does not cover correlations from external information between independent sessions.

The confidentiality endpoint is the client-released upload surface, consisting of the sanitized instruction and screenshot examined individually or as an ordered session history. The released-view audit applies a fixed reader to one sanitized instruction or screenshot. The five-turn audit supplies, in order, the sanitized instruction and screenshot from each of five turns in one Full session. This input contains every released reference and its visible type, root equality, projection, layout, and optional length cues. These audits exclude client-private state, restored payloads, and cloud-side context not contained in the released views. If a restored value or action result later appears in an uploaded instruction or screen, it becomes a new source occurrence and passes through SAN. Local resolution does not retroactively alter an upload already observed by the planner or a reader.

\subsection{Security and Functional Goals}
\label{subsec:goals}

\paragraph{G1. Source-linked removal.}
SAN seeks to remove the raw source content of every accepted occurrence from the final released instruction or screen. Its source trace links the occurrence to the prompt offsets or screenshot regions rewritten before upload. Recognition, source-tracing, or rendering failures remain possible and are measured, not assumed away.

\paragraph{G2. Reference continuity with explicit release cost.}
REF seeks to preserve a locally resolvable task path after SAN removes a protected value. Compatible appearances may reuse a session root, and supported address components may be requested through projections. Under the root-scoped contract, an active address root may select any supported component in its local record, including one not previously rendered as a projection. This capability is evaluated by reference-mediated execution. Released-view audits charge only uploaded information. Full exposes structured reference and layout cues absent from the SAN-only view. These cues can support task execution but may also help reconstruct protected values. REF provides no identity oracle. Its deterministic resolution contract applies to the root selected during binding, while agreement with the intended referent remains a measured outcome.

\paragraph{G3. Session-confined local resolution.}
Resolution is confined to roots issued in the active session and to values or components stored in their typed entity records. A model-generated string does not create a root, component, or authorization state. Failed requests remain unresolved and release no local value through the resolution interface. This contract controls local resolution rather than the semantic appropriateness of the surrounding action. For fixed issued session state and returned text, resolution is independent of unrelated values or aliases in the local privacy repository.

\subsection{Evaluation Views}

A \emph{SAN-only} view is erase-only. Detected source text or pixels are removed, and the renderer adds no readable replacement metadata. The evaluation includes SAN-only and Full variants of AbI and PrivacyGuard. The evaluated Presidio configurations are natively SAN-only and have no REF stage.

SAN is assessed through pre-render source-occurrence coverage and residual-character diagnostics, with fixed-reader source-value exposure from each final instruction or screen. On CrossTurnPII, REF is assessed through reference-mediated task capability and, for $m\in\{\mathrm{AbI},\mathrm{PrivacyGuard}\}$, the method-matched conditional re-exposure rate
\begin{equation}
\mathrm{RER}(m,A)
=
\frac{
\sum_i
\mathbf 1\{S_{m,A,i}=0\land F_{m,A,i}=1\}
}{
\sum_i
\mathbf 1\{S_{m,A,i}=0\}
},
\label{eq:method-matched-ref-increment}
\end{equation}
where $S_{m,A,i}$ and $F_{m,A,i}$ indicate whether fixed reader $A$ recovers gold entity $i$ from method $m$'s SAN-only and Full views. If the denominator is zero, $\mathrm{RER}(m,A)$ is undefined and reported as not applicable. The common-protected-set comparison evaluates each method only on entities unrecovered from both methods' SAN-only views, so the same entity set constrains both numerator and denominator.

For Full, \emph{audited source-value leakage} is the deduplicated union of values exposed exactly or partially in individual released views and reconstructed by the five-turn reconstruction probe. Task success remains a separate endpoint. Appendix~\ref{sec:appendix-metrics} defines the remaining metrics and scorers.

\section{Design and Implementation}
\label{sec:method}

PrivacyGuard first reconstructs source-linked character streams from the instruction and OCR output. The detector produces typed candidate evidence with provisional source links. Shared adjudication resolves competing candidates, and type-specific parsing finalizes each surviving candidate as an accepted occurrence. Every accepted occurrence contains a validated source trace and a typed entity record. SAN and REF consume this shared output and do not independently accept candidates. SAN uses the source trace to rewrite the source extent before upload. REF uses the typed entity record to maintain session identity, render structured references, and resolve returned references against local state.

\subsection{SAN Source-Linked Upload Sanitization}
\label{subsec:method-san}

\paragraph{Geometry-aware OCR stream reconstruction.}
Characters from the instruction retain their original offsets. OCR characters retain their block identifiers, bounding boxes, and in-block positions. Raw OCR order may separate one value when it wraps onto another line or spans neighboring boxes. PrivacyGuard computes two geometric successor relations to reconstruct this flow. A height-consistent rightward relation links fragments on the same line. A left-aligned downward relation links the end of one line to the beginning of the next. PrivacyGuard follows rightward successors before moving downward. A hard OCR break separates unrelated chunks, and a softer inline gap separates fragments within one reconstructed chunk. Ambiguous successors and noisy cycles are resolved by a deterministic reading order. The evaluated configuration fixes the relation predicates, geometric thresholds, and tie-breaking rules.
Reconstruction preserves the source links. The joined stream can be interpreted as one value, and every character points to the instruction offset or screenshot pixels that produced it.

\paragraph{Typed evidence and shared adjudication.}
The checkpoint-free clue scanner processes the reconstructed streams using validated formats, field labels, local context, optional local-privacy-repository matches, and active-session evidence. Each hypothesis records a proposed source extent, an inferred type, a local role, an evidence tier, provenance, and provisional source links. Labels and keys can support a neighboring value, but they are not promoted as values themselves. Hard breaks and negative scopes stop evidence from propagating into unrelated text. Numeric or alphanumeric evidence requires a private role assigned by local context or session evidence before it can become an accepted occurrence. The adjudicator delays acceptance when candidates overlap and first considers strict containment, type-specific conflicts, and the relative strength of hard and soft evidence. Type priority, score, and provenance resolve the remaining cases. Before discarding a losing candidate, the adjudicator removes the overlapping portion and revalidates any remaining source extent when possible. This step preserves a supported residual candidate without retaining an unsupported union. Each surviving candidate then undergoes type-specific parsing and evidence validation. A candidate becomes an accepted occurrence only when this stage produces a valid source trace and typed entity record.

\paragraph{Source-extent final-view rewriting.}
For each accepted occurrence, SAN rewrites the complete source extent linked by its validated source trace. Instruction spans are selected longest first and edited in reverse offset order, which preserves recorded source locations as earlier edits are applied. For screenshots, SAN maps the source trace back to one or more OCR boxes, refines the covered glyph boundary, and erases the linked pixels. In Full, REF renders the structured reference within the recorded geometry. An occurrence that spans several boxes retains their order, and the reference is split or reflowed when needed. SAN applies these rewrites to the final uploaded text and pixels. Failures in recognition, candidate acceptance, boundary refinement, erasure, or rendering can leave source content visible. The evaluation measures resulting recovery under the declared fixed readers and does not claim a reader-independent census of each residual glyph.

\subsection{Shared Structured Interpreter for Source Traces and Entity Records}
\label{subsec:method-shared-parser}

For each accepted occurrence, the structured interpreter emits a validated source trace and a typed entity record. The source trace records the accepted source extent together with its instruction offsets or screenshot geometry, while the typed entity record stores the type-specific identity and component state used by REF. A source-linked fragment also retains a \emph{display record} describing the component levels and geometry visible in the current appearance. SAN consumes the source extent and validated trace. REF uses the typed entity record and display record for binding, state accumulation, and reference rendering. Name records retain family-name, given-name, alias, and title evidence. Organization and scalar records retain type-appropriate canonical fields. Address records retain administrative, street, point-of-interest, building, and unit components together with admissible interpretations of uncertain administrative tokens. Addresses are the main evaluated component and projection case. The remaining PII types retain their own structured state.

\subsection{REF Structured Entity and Reference Lifecycle}
\label{subsec:method-ref}

\subsubsection{Same-Entity State Construction}

\paragraph{Current-turn appearance assembly.}
On screenshots, adjacent OCR address fragments can form one \emph{appearance group}. Grouping requires geometric adjacency and an allowed succession between their parsed component levels. Group members cannot occupy conflicting levels or duplicate the same primary surface. A group can span several OCR boxes. The appearance group produces one aggregate typed entity record for root lookup. This aggregate is bound once, after which all group members inherit the selected root. Each member retains its own source trace and display record. A fragment that does not enter an appearance group remains an independent accepted occurrence. For example, \emph{5th Ave.} and \emph{Apt.\ 24B} can share one root without merging their source pixels. Their separate display records also prevent either fragment from exposing the component labels of the other fragment.

\paragraph{Compatibility and deterministic root selection.}
REF processes only typed entity records produced for accepted occurrences. Scalar types reuse a root only after canonical equality. Names, organizations, and addresses use type-specific compatibility over their typed entity records. Compatibility requires the same entity type, no observed type-specific conflict, and a type-specific positive match. For scalar types, a positive match is canonical equality. For structured types, it is shared parser-recognized identity or component evidence admitted by the corresponding compatibility rules. Address comparison checks strong identity and administrative evidence together with street-level and detail evidence. Missing evidence is treated as unknown. Ambiguous administrative tokens are tested against their admissible interpretations. Compatibility rules out observed conflicts. It does not prove identity when the current observations lack distinguishing evidence. A false merge can occur when several distinct roots remain compatible and the selected root is not the intended referent. A false split can occur when a true continuation fails compatibility and receives a fresh root. Binding agreement and chain purity are measured outcomes. Proposition~\ref{prop:binding-separation} in Appendix~\ref{sec:appendix-analytical} states the implemented reuse constraint formally, and Appendix~\ref{sec:appendix-reference} gives the complete compatibility cases.

A binding input is either an independent accepted occurrence or the aggregate record of an appearance group. Let $c$ denote such a binding input, let $O(c)$ denote the normalized representation of its typed entity record, let $O(e)$ denote the corresponding representation of prior root $e$, and let $\operatorname{Compat}_{\tau}$ denote the compatibility predicate for type $\tau$. Candidate roots follow a deterministic type-specific order. Address lookup uses component indexes, including a canonical-road bucket, and evaluates only same-type roots with the production compatibility predicate. Root selection is
\begin{equation}
\begin{aligned}
    \operatorname{Bind}&(c,H)  = \\
    &\operatorname{first}_{\prec(c,H)}\left\{e\in H_{\tau(c)}\middle|\operatorname{Compat}_{\tau(c)}(O(c),O(e))\right\},
\end{aligned}
\label{eq:deterministic-binding}
\end{equation}
where an empty compatible set produces a fresh root. Within each index tier, $\prec(c,H)$ orders compatible roots by recency. The binder uses only this observable compatibility state. The gold intended referent is an evaluation label and is not available during selection.

Appendix~\ref{sec:appendix-analytical} compares recent-first with a counterfactual that allocates a fresh root when several prior roots remain compatible. On the continuation domain defined there, recent-first preserves every exact continuation preserved by that counterfactual. It may preserve a continuation when the most recent compatible root is intended. Recent-first can still select the wrong root when another compatible root is intended. This comparison concerns exact continuation on the stated domain and does not establish a lower false-merge risk. Binding inputs follow a fixed within-turn order. Each completed assignment updates the current-turn index, so later independent inputs can use roots created or updated earlier in the same turn. The fixed order makes same-turn updates deterministic.

\paragraph{Accumulated entity state.}
When an address reuses a root, PrivacyGuard merges compatible evidence into its accumulated entity state. It retains the more complete form of an equivalent component and deduplicates point-of-interest entries. It extends compatible detail sequences. A typed unit field can absorb redundant generic detail containing the same number tokens. PrivacyGuard then rebuilds the accumulated component order and canonical surface. The canonical surface is the current stored value returned by a bare root. Other PII types retain their type-compatible state and recorded source appearances. Local resolution reads only from this accumulated entity state.

\subsubsection{Reference Rendering and Local Resolution}

The root \texttt{[[TYPE\#N]]} identifies the accumulated entity state. Its suffix describes the component levels visible in the current appearance. Let $\mathcal K$ denote the fixed order \texttt{COUNTRY}, \texttt{PROV}, \texttt{CITY}, \texttt{DIST}, \texttt{ROAD}, and \texttt{DTL}. Let $\mathcal O(c)$ denote the ordered parsed components of occurrence $c$, and let $\ell(a)$ denote the display level of component $a$. PrivacyGuard constructs
\begin{equation}
\operatorname{SPEC}(c)
=
\operatorname{Ord}_{\mathcal K}
\left\{
\ell(a)
\middle|
a\in\mathcal O(c),
\ell(a)\neq\emptyset
\right\}.
\label{eq:projection-construction}
\end{equation}
A road-only fragment can render as \texttt{[[ADDR\#2.ROAD]]} even when the root stores a complete destination. \texttt{SPEC} describes the released appearance. It does not restrict the root-scoped projection contract. The return path can resolve any supported component present in the accumulated entity state, including one not shown in the current suffix. Generic NUM and ALNUM references may include an implementation-normalized character count in \texttt{LEN=L}. This cue preserves coarse format information such as the length of a numeric code or alphanumeric identifier without exposing the raw value. \texttt{LEN=L} is not a selector and cannot affect local resolution. Appendix~\ref{sec:appendix-reference} gives the exact normalization rules. NAME components remain local binding evidence in the evaluated configuration and are not exposed as external projections.

The return path scans complete double-bracketed candidates containing \texttt{\#}. It never joins text across bracket boundaries. Normalization folds only allowlisted full-width characters and separator confusables, removes supported internal whitespace, and matches the complete result against the anchored reference grammar.

Resolution requires an active-session root whose stored type matches the request. A bare root requests the canonical surface in the accumulated entity state. An address suffix requests a supported component that is locally present. A NUM or ALNUM \texttt{LEN=L} suffix carries descriptive metadata. Its nonnegative integer value is parsed but is not compared with the stored value. The resolver does not synthesize a missing component or widen a failed projection into a request for the complete value. A valid request restores the stored value or component locally, and every other candidate remains literal. Restoration produces the resolved action payload consumed by the trusted executor without modifying the upload sequence already released to the planner. Appendix~\ref{sec:appendix-algorithms} gives the binding and resolution pseudocode. Appendices~\ref{sec:appendix-reference} and~\ref{sec:appendix-restore-cases} give the exact projection and validation cases.

\subsection{Architecture and Configuration}
\label{subsec:method-config}

The pre-upload mediation hook performs OCR reconstruction, clue generation, shared adjudication, type-specific parsing, source rewriting, session binding, and reference rendering. The return hook runs after planner output and before the trusted executor. It validates and resolves returned references against the active session mapping store. Both hooks remain inside the trusted mediation tier. The evaluated implementation connects these hooks to geometry-aware OCR reconstruction, a checkpoint-free detector, type-specific parsers, session binding, reference rendering, and a restorer behind a fixed interface. Appendix~\ref{sec:appendix-topology} describes the server-hosted comparison profile and the device-resident handset profile. The implementation uses 33 static lexicon assets that occupy 0.28\,MiB and loads no learned named-entity-recognition checkpoint. 
OCR-free controls isolate the detector and structured interpreter from OCR behavior. Shared-OCR controls hold the frontend fixed between methods. Native-system rows retain each complete released pipeline.

The default release configuration uses a memory-only mapping store. It retains entity records between turns and removes an idle session after a configurable timeout. The experiments use a timeout of 3,600 seconds, and explicit session deletion uses the same store interface. An opt-in evaluation and audit mode can persist buffered session mapping state to local JSON. The persisted state includes issued roots and their local typed entity records and can contain sensitive value state. This mode is separate from the memory-only default and remains inside the trusted local boundary. The upload-confidentiality claim does not cover compromise or secure deletion of this audit file. Raw mapping values remain in trusted local state. Only the rewritten view and, in Full, its structured references cross the confidentiality boundary. RQ1 measures source content that remains visible in the released view.

\section{Evaluation Methodology}
\label{sec:evaluation-methodology}

\subsection{Research Questions and Evaluation Logic}
\label{subsec:eval-questions}

\textbf{RQ1} asks how completely SAN removes source-linked PII, what visible residuals remain, and how erase-only sanitization affects value-independent released-view utility. \textbf{RQ2} asks whether REF preserves entity continuity and creates a locally resolvable task path after SAN removes a required value, and measures the reconstruction exposure introduced by released references. \textbf{RQ3} measures the privacy exposure and task success of the complete Full pipelines on controlled and external workloads. \textbf{RQ4} measures latency, memory, and completion under device-resident execution. RQ1 and RQ2 are the empirical counterparts of the SAN and REF objectives, while RQ3 reports their realized Full-pipeline outcomes. Appendix Table~\ref{tab:dataset-evidence-roles} maps each evidence source to these questions and states its scope.

\subsection{Evidence Sources and Workloads}
\label{subsec:eval-evidence}

The evaluation combines controlled, external, targeted, and deployment evidence. The CrossTurnPII controlled corpus supplies exact bilingual values and source regions, text--image correspondence, five-turn entity identity, and component projections. It contains 4,800 prompts and 2,400 screenshots in 480 five-turn sessions. Its 800-image reader subset contains 5,632 gold entities and is selected without reader outputs by retaining complete personas and all four sessions associated with each selected persona. The structured-address stress split contains 7,800 forms over 240 address identities and varies completeness, order, abbreviation, OCR noise, and multi-block fragmentation. PrivScreen and AndroidLab provide external screenshot and task evidence. The fixed live-interface suite and randomized multi-root suite provide targeted REF capability and binding-policy evidence, while the Xiaomi 12S workload measures device-resident execution.

\paragraph{Construction independence and data handling.}
The controlled sources avoid collection and retention of participant or third-party PII. Author-directed model assistance supplied draft scenario wording, after which deterministic scripts instantiated values, spans, regions, entity links, and projections and produced the gold labels and measurements. The structured-address split reuses controlled persona records while inducing a distinct distribution over address forms. The randomized multi-root suite uses a separate generator, templates, and values. A type-aware audit checks all 33 detector lexicons against the evaluated records for exact full-name and high-specificity identifier overlap. The runtime detector and binder do not receive construction-derived source, identity, or projection labels. External benchmarks and the independent multi-root generator provide evidence outside the controlled rendering process. Appendix~\ref{sec:appendix-provenance} documents corpus construction, reader-subset selection, generator independence, lexicon auditing, and row-level provenance.

\subsection{Baselines and Comparison Protocols}
\label{subsec:eval-baselines}

\paragraph{Native system comparisons.}
AbI is the closest prior full-pipeline comparator. Its native detector output feeds deterministic placeholders, value-map reuse with fuzzy OCR alignment, and whole-value local restoration. The evaluated Presidio AnalyzerEngine (\textit{PrA}) and GLiNER detector exposed through Presidio's recognizer interface (\textit{PrG}) configurations render detected screenshot regions with an opaque fill and no readable replacement string, so their outputs are SAN-only. They do not provide a reference lifecycle and do not enter REF or Full comparisons. Tables label PrivacyGuard as \textit{Ours}. Native-system rows retain each system's complete evaluated frontend, rendering behavior, and lifecycle and report the resulting pipeline output.

\paragraph{Controlled SAN comparisons.}
The AbI SAN-only adapter retains AbI's native detector and detected regions but replaces its token overlay with an opaque fill and releases no token metadata. OCR-free inputs isolate the detector and structured interpreter from OCR behavior. Shared-OCR inputs give identical OCR records to every method and isolate source coverage, residual leakage, precision, and grounding from frontend differences. Native comparisons remain separate from these controls and retain each system's declared OCR frontend. Appendix~\ref{sec:appendix-baseline-audit} records the source snapshots, adapters, and mapping between configurations and estimands.

\paragraph{Matched lifecycle ablation.}
The PrivacyGuard ablation uses the AbI-style exact-value-reuse interface, labeled \textit{Exact} as a standalone condition and \textit{+Exact} as a method suffix in tables. It fixes PrivacyGuard's SAN candidates, typed records, renderer, reference syntax, and local restorer, while replacing structured same-entity binding with exact canonical-value reuse and disabling component projections. This comparison isolates the continuity and task effects of the structured reference lifecycle without changing the source-removal surface or local restoration implementation.

\paragraph{Execution profiles.}
All compared pipelines within each server-hosted experiment execute on the same RTX 4090 server to hold hardware placement constant. This server belongs to the trusted experimental apparatus, remains distinct from the remote planner, and does not constrain the intended placement of PrivacyGuard. Separately, a 16-task handset profile runs the packaged endpoint on a Xiaomi 12S, released in 2022, with a Snapdragon 8+ Gen 1 processor and 12\,GB RAM. The handset ran Android 15 during evaluation. This profile measures device-resident latency, memory, persistence, and completion. Appendix~\ref{sec:appendix-topology} records the server-hosted comparison profile and the handset profile.

\subsection{Endpoints, Readers, and Statistics}
\label{subsec:eval-protocol}

\emph{Primary system endpoints.}
For Full, the primary privacy endpoint is audited source-value leakage, defined as the deduplicated union of exact or partial per-view exposure and exact or partial five-turn reconstruction from the ordered Full uploads. The primary task endpoint is Full task success on AndroidLab. Privacy and task utility remain separate outcomes.

\emph{Mechanism endpoints.}
RQ1 measures SAN complete-source coverage and residual leakage before readable references are rendered. GenR measures coverage of gold source occurrences at the declared threshold, while $\mathrm{LR}_{\mathrm{char}}$ and $\mathrm{LR}_{\mathrm{detect}}$ measure residual source content in text and screenshot settings. Placeholder-free SAN-only views provide the rendered-view exposure comparison without charging readable reference content to SAN. RQ2 measures strict cross-form continuity, automatically verified reference-mediated completion, conditional re-exposure among SAN-protected entities, and targeted reconstruction from released references. The common-protected-set comparison aligns AbI and PrivacyGuard on gold entities protected by both SAN views. AndroidLab's SAN-only-to-Full contrast supports RQ2 mechanism attribution, while its complete Full-system comparisons are reported under RQ3.

\emph{Diagnostics.}
Det-P, BBoxR, and Mod characterize candidate precision, geometric grounding, and rewrite count. Chain reuse and projection fidelity diagnose structured continuity and rendering. Malformed-reference cases test the local-resolution boundary, while component timing and memory measurements characterize deployment cost. These diagnostics support mechanism attribution but are not substitutes for the Full privacy and task endpoints. Appendix~\ref{sec:appendix-metrics} gives the complete definitions and formulas.

CrossTurnPII and PrivScreen use two hosted readers from different model families, \texttt{qwen3.7-plus} (\textit{Qwen}) and \texttt{gemini-3-flash-preview} (\textit{Gemini}). Both use temperature 0. The first successful completion is cached for each predefined evaluation input tuple, which includes the reader, prompt, released-view input, and target when applicable. Within each audit, the reader configuration, prompts, retry policy, and scorer are fixed across the compared protection methods. CrossTurnPII includes released-view PII extraction, targeted single-reference reconstruction from Full, and targeted five-turn reconstruction from ordered Full uploads. PrivScreen retains its benchmark field-query privacy questions and normal visual questions and is not used to estimate the RQ2 REF increment. Appendix~\ref{sec:appendix-screenshot} reports the complete released-view results and diagnostic breakdowns. 

\paragraph{Uncertainty protocol.}
RQ1 prompt and structured-address intervals use 5,000 percentile-bootstrap resamples of complete personas with PRNG seed 1729, preserving correlated turns and variants. RQ2 method-matched conditional re-exposure rates, common-protected-set differences, and signed Full-minus-SAN-only contrasts use 10,000 persona-paired resamples. The signed contrasts remain secondary diagnostics. The same 10,000 persona-paired protocol estimates the RQ3 PrivacyGuard-minus-AbI common-gold audited-leakage differences. Each paired resample uses the same sampled personas for both compared methods. AndroidLab task-success differences use 10,000 task-paired resamples of the same 138 tasks, with the same sampled tasks used for both compared conditions in each resample. PrivScreen reports the fixed-reader absolute SAN and Full outcomes on its common benchmark questions rather than an estimated REF increment. Controlled mechanism, local-resolution, binding-conformance, and live-interface suites are reported as fixed-suite counts. The reported intervals condition on the cached reader outputs and retained planner trajectories. They quantify corpus- or task-resampling uncertainty rather than repeated hosted-model variability. Appendix~\ref{sec:appendix-statistics} gives the complete intervals, estimand mappings, and grouping keys.

\section{Evaluation Results}
\label{sec:evaluation}

\subsection{RQ1: SAN Source Removal and Released-View Effects}
\label{subsec:eval-text-address}

RQ1 evaluates whether SAN removes complete source-linked PII occurrences from text and screenshots. We first examine controlled text and structured-address inputs, then test rewritten screenshot pixels on CrossTurnPII and PrivScreen. Matched mechanism analyses explain the observed source-removal differences, while SAN-only released views measure value-independent visual utility.

\subsubsection{Controlled Source Removal}

\paragraph{Text and structured-address coverage.}
Text inputs isolate SAN before OCR and rendering. The structured-address stress split additionally varies completeness, order, abbreviation, OCR noise and multiblock fragmentation. Table~\ref{tab:eval-text-address-main} shows complete-source coverage and source-character leakage.

\begin{table}[t]
\centering
\caption{SAN source removal on controlled text and structured-address inputs.}
\label{tab:eval-text-address-main}
\begin{tabularx}{\linewidth}{@{}C{0.8}C{1.8}C{0.85}C{0.85}C{0.85}C{0.85}@{}}
\toprule
& & \multicolumn{2}{c}{English} & \multicolumn{2}{c}{Chinese} \\
\cmidrule(lr){3-4}\cmidrule(lr){5-6}
Input & Method &
GenR\newline @1.0$\uparrow$ &
$\mathrm{LR}_{\mathrm{char}}$\newline $\downarrow$ &
GenR\newline @1.0$\uparrow$ &
$\mathrm{LR}_{\mathrm{char}}$\newline $\downarrow$ \\
\midrule
\multirow{2}{*}{Text}
& Best baseline
& 0.7912 & 0.1283 & 0.5048 & 0.3998 \\
& Ours
& \best{0.9740} & \best{0.0125}
& \best{0.9971} & \best{0.0023} \\
\midrule
\multirow{4}{=}{Address\newline stress}
& AbI
& 0.3448 & 0.5094 & 0.0986 & 0.8921 \\
& PrA
& 0.1160 & 0.3968 & 0.1944 & 0.6628 \\
& PrG
& 0.6247 & 0.0537 & 0.6567 & 0.0114 \\
& Ours
& \best{0.8981} & \best{0.0430}
& \best{0.9572} & \best{0.0071} \\
\bottomrule
\end{tabularx}
\PGTableNote{Text baseline: PrG (EN), PrA (ZH). Appendix~\ref{sec:appendix-prompt} gives complete baselines and intervals.}
\end{table}

PrivacyGuard improves complete-source coverage and lowers residual leakage in both languages. On structured addresses, residual leakage falls from AbI's 0.5094 and 0.8921 to 0.0430 and 0.0071 in English and Chinese. Complete-address coverage rises from 0.3448 and 0.0986 to 0.8981 and 0.9572. In the text and address AbI comparisons reported in the main and appendix tables, residual leakage falls by 91.56\%--99.68\%. The complete grouped-bootstrap intervals appear in Appendix Table~\ref{tab:appendix-bootstrap-summary}. Type-stratified results retain the same main direction for conventional PII. Service identifiers produce a more varied ordering, with Presidio+rules leading the controlled English cell and PrivacyGuard leading the controlled Chinese and external PrivScreen cells. Appendix Tables~\ref{tab:appendix-prompt-per-type} and~\ref{tab:appendix-privscreen-per-type} report the complete strata.

\paragraph{Complete versus partial coverage.}
GenR@0.8 measures high-overlap coverage, while GenR@1.0 requires the complete gold source occurrence. On structured English prompts, AbI falls from 0.8212 at 0.8 to 0.5456 at 1.0. PrivacyGuard remains at 0.9794 and 0.9740. Its Chinese result is 0.9971 at both thresholds. These narrow gaps show that most covered occurrences in the controlled inputs are complete, not partial.

\subsubsection{Rendered Screenshot Removal}
\label{subsec:eval-screenshot}

Table~\ref{tab:eval-screenshot-san-main} reports source removal on controlled CrossTurnPII and external PrivScreen screenshots. Native results use each system's own frontend. Shared-OCR columns provide identical OCR records to all methods and isolate the downstream detector, interpreter, and rewriting path. Every endpoint is computed before REF renders readable references.

\begin{table}[t]
\centering
\caption{SAN source removal on controlled CrossTurnPII and external PrivScreen screenshots.}
\label{tab:eval-screenshot-san-main}
\begin{tabularx}{\linewidth}{@{}C{0.45}L{0.75}C{1.1}C{1}C{1.6}C{1.1}@{}}
\toprule
Lang. & Method &
Native\newline $\mathrm{LR}_{\mathrm{detect}}\downarrow$ &
Shared\newline Det-P$\uparrow$ &
Shared\newline GenR@1.0$\uparrow$ &
Shared\newline $\mathrm{LR}_{\mathrm{detect}}\downarrow$ \\
\midrule
\multicolumn{6}{@{}l}{\textit{Dataset: CrossTurnPII}} \\
EN & AbI
& 0.1850 & 0.3249 & 0.5521 & 0.3397 \\
EN & PrA
& 0.1733 & 0.4313 & 0.5629 & 0.0482 \\
EN & PrG
& 0.4454 & 0.4542 & 0.5486 & 0.3071 \\
EN & Ours
& \best{0.0316} & \best{0.8125} & \best{0.9113} & \best{0.0316} \\
\cmidrule(lr){1-6}
ZH & AbI
& 0.5673 & \best{0.9802} & 0.3720 & 0.5657 \\
ZH & PrA
& 0.4509 & 0.7545 & 0.4971 & 0.2576 \\
ZH & PrG
& 0.6632 & 0.5982 & 0.3743 & 0.6001 \\
ZH & Ours
& \best{0.0389} & 0.8258 & \best{0.9229} & \best{0.0389} \\
\midrule
\multicolumn{6}{@{}l}{\textit{Dataset: PrivScreen}} \\
-- & AbI
& 0.2018 & 0.6611 & 0.6743 & 0.1582 \\
-- & PrA
& 0.2131 & 0.4820 & 0.4200 & 0.1779 \\
-- & PrG
& 0.4947 & 0.4842 & 0.6705 & 0.4750 \\
-- & Ours
& \best{0.0724} & \best{0.9122} & \best{0.8600} & \best{0.0724} \\
\bottomrule
\end{tabularx}
\PGTableNote{Native=native OCR frontend; Shared=common PP-OCRv5-Mobile records. Appendix~\ref{sec:appendix-screenshot} gives full results.}
\end{table}

On CrossTurnPII, PrivacyGuard has the lowest native residual leakage and highest shared-OCR complete-source coverage in both languages. PrivScreen shows the same trend externally. PrivacyGuard lowers native residual leakage to 0.0724 and reaches shared-OCR GenR@1.0 of 0.8600. Its PrivScreen Det-P is 0.9122 and BBoxR@0.5 is 0.9452. With the same OCR input, modified regions per changed image fall from AbI's 2.5224 to 2.3116. On PrivScreen, higher coverage and precision with fewer modified regions indicate more focused rewriting.

\subsubsection{Mechanism Attribution and Released-View Effects}

\paragraph{Matched SAN mechanisms.}
Table~\ref{tab:eval-san-mechanisms-main} changes one SAN mechanism at a time under fixed inputs. Removing shared adjudication reduces English-address GenR@1.0 by 0.5666. Removing OCR-stream reconstruction reduces PrivScreen GenR@1.0 by 0.1381 and Det-P by 0.2362. Enrolling first-pass misses in the optional local privacy repository raises English and Chinese prompt GenR@1.0 to 0.9867 and 1.0000. All cross-system headline results keep the repository empty. The split-block diagnostic confirms that the shared interpreter can retain linked source extents and one-root chains when addresses span several OCR regions. Exact-value lifecycles cannot reuse one root for fragments with different canonical surfaces. Appendix Table~\ref{tab:eval-address-split-main} reports the complete language and block-count results. Appendix Table~\ref{tab:appendix-screenshot-ocr-matrix} reports sensitivity to four OCR backends.

\begin{table}[t]
\centering
\caption{Matched SAN mechanism tests with same inputs.}
\label{tab:eval-san-mechanisms-main}
\begin{tabularx}{\linewidth}{@{}L{1}L{2}C{0.5}C{0.5}@{}}
\toprule
Mechanism & Endpoint & Without & With \\
\midrule
Adjudication
& EN address GenR@1.0$\uparrow$
& 0.3315 & \best{0.8981} \\
\multirow{2}{=}{OCR recon.}
& PrivScreen GenR@1.0$\uparrow$
& 0.7219 & \best{0.8600} \\
& PrivScreen Det-P$\uparrow$
& 0.6760 & \best{0.9122} \\
\multirow{2}{=}{Local repo.}
& EN prompt GenR@1.0$\uparrow$
& 0.9740 & \best{0.9867} \\
& ZH prompt GenR@1.0$\uparrow$
& 0.9971 & \best{1.0000} \\
\bottomrule
\end{tabularx}
\PGTableNote{The local-repository rows are matched before/after results; all headline comparisons keep the repository empty.}
\end{table}

\paragraph{SAN-only exposure and visual utility.}
Removing readable references isolates the released-view effect of SAN. On CrossTurnPII, PrivacyGuard SAN-only exposure is 0.1289 and 0.0437 under Qwen for English and Chinese, and 0.1332 and 0.0532 under Gemini. AbI SAN-only exposure is 0.2775 and 0.7513 under Qwen, and 0.2800 and 0.7598 under Gemini. On PrivScreen, PrivacyGuard SAN-only exposure is 0.0585 under Qwen and 0.0576 under Gemini, compared with 0.1992 and 0.2096 for AbI. PrivacyGuard's visual-answer accuracy is 0.916 and 0.994. Relative to the original views, these values differ by $+0.4$ and $+2.6$ percentage points. The benchmark questions request controls and navigation targets rather than protected values. These results measure value-independent released-view utility and do not imply that erase-only SAN preserves tasks that require a protected value. Appendix Tables~\ref{tab:appendix-crossturnpii-complete} and~\ref{tab:appendix-privscreen-complete} retain the complete SAN and Full rows.

RQ1 shows that SAN improves complete-source removal on controlled and external inputs, while erase-only visual utility remains close to the original views on PrivScreen’s value-independent questions.

\subsection{RQ2: REF Capability, Continuity, and Reconstruction Cost}
\label{subsec:eval-reference}

RQ2 evaluates whether REF creates a locally resolvable task path after SAN removes a required value, whether its structured lifecycle preserves changed appearances, and how much reconstruction exposure the released references introduce.

\subsubsection{Verified Reference-Mediated Utility Recovery}

\paragraph{Reference-required tasks.}
The fixed live-interface suite contains tasks selected in advance to require a protected phone number, full address, or address component. A run succeeds only when the final action submits one trajectory-issued target reference with the required projection and local restoration produces the gold value or component. Plaintext guesses do not count. Under this reference-required scoring rule, SAN-only has a structural value of zero because it releases neither the raw target nor a trajectory-issued reference. Full completes 17/30 Qwen runs and 20/30 Gemini runs. Exact-value reuse completes 8/30 and 7/30. On the 18 component-only runs per planner, exact-value reuse completes none, while Full completes 6 under Qwen and 12 under Gemini.

\begin{table}[t]
\centering
\caption{REF capability and continuity on separate workloads. Panels use different denominators and are not pooled.}
\label{tab:eval-ref-capability-continuity}

\textit{Panel A. Strictly verified task success}

\begin{tabularx}{\linewidth}{@{}L{1}C{1}C{1}C{1}C{1}@{}}
\toprule
Planner & SAN & Local & Exact & Full \\
\midrule
Qwen
& 0  & 0/30 & 8/30 & \best{17/30} \\
Gemini
& 0  & 1/30 & 7/30 & \best{20/30} \\
\bottomrule
\end{tabularx}

\textit{Panel B. Cross-form address continuity}

\begin{tabularx}{\linewidth}{@{}L{1.1}C{0.9}C{0.9}C{1.05}C{1.05}@{}}
\toprule
Method & EN PCR & ZH PCR & EN Chain & ZH Chain \\
\midrule
AbI
& 0.1465 & 0.0000 & 0.0000 & 0.0000 \\
PrA+Exact
& 0.0346 & 0.0074 & 0.0000 & 0.0000 \\
PrG+Exact
& 0.0499 & 0.0074 & 0.0000 & 0.0000 \\
Ours+Exact
& 0.0361 & 0.0486 & 0.0000 & 0.0000 \\
Ours Full
& \best{0.8084} & \best{0.8175}
& \best{0.8708} & \best{0.6021} \\
\bottomrule
\end{tabularx}
\PGTableNote{Panel A uses ten templates with three runs each. Local denotes turn-local references. SAN is structurally zero. Success requires an issued root, required projection, and correct local resolution. Strict PCR treats unbound appearances as inconsistent. Chain averages two-/three-fragment reuse.}
\end{table}

\paragraph{AndroidLab utility recovery.}
PrivacyGuard performs 3{,}527 occurrence-level replacements in its AndroidLab trajectories. Starting from SAN-only, Full raises Qwen3.7-Plus SR from 50.00\% to 56.52\% and Gemini-3-Flash-Preview SR from 60.14\% to 63.04\%. The paired changes are $+6.52$ points with 20 gains, 11 losses, and 107 ties, and $+2.90$ points with 11 gains, 7 losses, and 120 ties. Their 95\% intervals are $[-1.45,+14.49]$ and $[-2.90,+8.70]$. The point estimates and strictly verified tasks support utility recovery for a subset of tasks, but both AndroidLab intervals include zero.

\paragraph{Local-resolution contract.}
The resolver restores every evaluated valid root, projection, history-reuse case, and supported normalized variant and rejects every evaluated unsupported or malformed input without releasing a raw value. Appendix Table~\ref{tab:eval-restore-main} and Appendix~\ref{sec:appendix-restore-cases} report the coverage families and complete rejection cases.

\subsubsection{Cross-Form Continuity and Binding}

The 7,800-form address split evaluates cross-system continuity and the matched lifecycle ablation. PrivacyGuard reaches strict PCR@0.80 of 0.8084 in English and 0.8175 in Chinese. The matched exact-value-reuse ablation reaches 0.0361 and 0.0486 with unchanged GenR and $\mathrm{LR}_{\mathrm{char}}$. Its chain reuse is zero. Full retains substantially higher two- and three-block chain continuity because fragments with different canonical surfaces can reuse one structured entity root. Panel B of Table~\ref{tab:eval-ref-capability-continuity} reports the aggregate results, and Appendix~\ref{sec:appendix-address} reports the complete factorization and projection analysis.

The 480 controlled sessions show 0/4{,}352 binder-selected root misattachments and 0/2{,}031 impure chains. In 6{,}000 generator-conditioned gold continuations, recent-first selects the intended root in 5{,}286 trials, versus 4{,}236 for fresh-on-ambiguity. Recent-first still selects a non-gold compatible root in 714 trials. These outcomes characterize policy--intent agreement under the declared generator and do not establish a general false-merge guarantee. Appendix~\ref{sec:appendix-binding-risk} reports exhaustive implementation conformance, the complete generator-conditioned table, and the dataset-conditioned policy analysis.

\subsubsection{Incremental Reconstruction Cost}

On each method's own SAN-protected population, Full newly exposes 260/2{,}724 entities for AbI and 133/5{,}148 for PrivacyGuard under Qwen. The corresponding Gemini counts are 215/2{,}693 and 139/5{,}109. These denominators reflect different SAN-protected populations. Panel A of Table~\ref{tab:eval-ref-reconstruction-cost} compares the methods on gold entities protected by both SAN views.

\begin{table}[t]
\centering
\caption{REF reconstruction cost on CrossTurnPII.}
\label{tab:eval-ref-reconstruction-cost}
\textit{(a) Conditional re-exposure on the common protected set}
\begin{tabularx}{\linewidth}{@{}C{0.7}C{0.4}C{0.8}C{0.4}C{0.4}C{3.3}@{}}
\toprule
Reader & Lang. & Protected & AbI (\%) & Ours (\%) & $\Delta$ [95\% CI] \\
\midrule

\multirow{3}{*}{Qwen}
& All
& 2,511
& 8.56
& \best{0.68}
& $-7.89$ [$-10.18,-5.58$] \\

& EN
& 1,836
& 10.95
& \best{0.60}
& $-10.35$ [$-13.11,-7.89$] \\

& ZH
& 675
& 2.07
& \best{0.89}
& $-1.19$ [$-2.79,+0.17$] \\

\multirow{3}{*}{Gemini}
& All
& 2,479
& 6.62
& \best{0.56}
& $-6.05$ [$-7.74,-4.26$] \\

& EN
& 1,828
& 8.42
& \best{0.66}
& $-7.77$ [$-9.86,-5.74$] \\

& ZH
& 651
& 1.54
& \best{0.31}
& $-1.23$ [$-2.08,-0.47$] \\

\bottomrule
\end{tabularx}

\textit{(b) Exact reconstruction from emitted references}

\begin{tabularx}{\linewidth}{@{}C{0.8}C{1.05}C{1.05}C{1.05}C{1.05}@{}}
\toprule
&
\multicolumn{2}{c}{Single view}
&
\multicolumn{2}{c}{Five turns} \\
\cmidrule(lr){2-3}
\cmidrule(lr){4-5}

Reader
& AbI
& Ours
& AbI
& Ours \\
\midrule

Qwen
& 0.0602
& \best{0.0130}
& 0.0622
& \best{0.0299} \\

Gemini
& 0.0716
& \best{0.0159}
& 0.0644
& \best{0.0273} \\

\bottomrule
\end{tabularx}

\PGTableNote{\emph{Note.} Panel (a) conditions on the common SAN-protected set. Panel (b) uses method-conditioned emitted-reference targets. Populations and denominators differ; panels are not pooled.}
\end{table}

On the common protected set, pooled conditional re-exposure is 0.68\% for PrivacyGuard and 8.56\% for AbI under Qwen. The corresponding Gemini rates are 0.56\% and 6.62\%. The pooled differences under both readers exclude zero. Among the four language-specific comparisons, three intervals exclude zero, while the Qwen Chinese interval includes zero. Exact reconstruction from PrivacyGuard references remains nonzero. Its pooled single-view rates are 0.0130 and 0.0159 under Qwen and Gemini. Its five-turn rates are 0.0299 and 0.0273. These rates are lower on the evaluated method-conditioned reference surfaces than the corresponding AbI rates, but they confirm that structured references introduce measurable reconstruction information. Binding correctness remains a separate outcome because changing only the gold intended-root label does not alter a fixed released surface or a fixed reader's reconstruction input.

RQ2 shows that REF preserves a locally resolvable task path and strongly improves cross-form continuity, while its released references retain measurable reconstruction cost.

\subsection{RQ3: Full-Pipeline Privacy and Task Outcomes}
\label{subsec:eval-complete}

RQ3 reports the complete Full pipelines. Section~\ref{subsec:eval-text-address} reports the corresponding SAN-only views. This subsection evaluates audited source-value leakage, final-view privacy and visual utility, and complete-pipeline task success. Privacy and task utility remain separate endpoints.

\subsubsection{Controlled CrossTurnPII Full Upload Views}

Table~\ref{tab:eval-full-leakage} reports the primary audited source-value leakage endpoint on identical common-gold sessions and denominators. The complete single-view SAN and Full rows appear in Appendix Table~\ref{tab:appendix-crossturnpii-complete}. PrivacyGuard has lower audited source-value leakage in all four reader--language cells, and every paired interval excludes zero. Relative to AbI, the reductions range from 33.55\% to 82.60\%. The five-turn reconstruction probe adds little after deduplication. Appendix Table~\ref{tab:appendix-full-union-decomposition} reports the per-view and five-turn union values directly. Appendix Table~\ref{tab:appendix-alr-monotone-sensitivity} additionally charges PrivacyGuard for any value exposed by either SAN-only or Full, and all four conservative sensitivity intervals remain below zero.

\begin{table}[t]
\centering
\caption{CrossTurnPII audited source-value leakage on common-gold five-turn Full uploads. $\Delta$ is Ours minus AbI.}
\label{tab:eval-full-leakage}
\begin{tabularx}{\linewidth}{@{}C{0.3}C{0.55}C{0.5}C{0.5}C{3.15}@{}}
\toprule
Lang. & Reader & AbI & Ours & $\Delta$ [95\% CI] \\
\midrule
EN & Qwen
& 0.3486 & \best{0.2080}
& $-0.1406$ [$-0.1698,-0.1080$] \\
EN & Gemini
& 0.3271 & \best{0.2174}
& $-0.1098$ [$-0.1429,-0.0750$] \\
ZH & Qwen
& 0.7022 & \best{0.1222}
& $-0.5801$ [$-0.6126,-0.5458$] \\
ZH & Gemini
& 0.7163 & \best{0.1426}
& $-0.5737$ [$-0.6049,-0.5431$] \\
\bottomrule
\end{tabularx}
\PGTableNote{$N=1{,}394$ per EN row and $N=1{,}424$ per ZH row. Audited leakage deduplicates per-view exposure and five-turn reconstruction. Intervals use paired persona resampling.}
\end{table}

\subsubsection{External PrivScreen Full Views}

PrivScreen evaluates final released views on a captured external screenshot distribution. Table~\ref{tab:eval-privscreen-full} compares the two Full pipelines. The complete Original, SAN-only, and Full configurations remain in Appendix Table~\ref{tab:appendix-privscreen-complete}. PrivacyGuard lowers exposure from 0.2134 to 0.0793 under Qwen and from 0.2229 to 0.0831 under Gemini. These correspond to reductions of 62.84\% and 62.72\%. Reader-matched visual accuracy differs from AbI by $-0.40$ and $+1.60$ percentage points. These questions measure the usability of the final released screenshot and do not exercise planner return or local restoration.

\begin{table}[t]
\centering
\caption{PrivScreen privacy exposure and visual-answer utility for the two Full pipelines.}
\label{tab:eval-privscreen-full}
\begin{tabularx}{\linewidth}{@{}L{1.5}C{0.97}C{0.78}C{0.97}C{0.78}@{}}
\toprule
& \multicolumn{2}{c}{Qwen}
& \multicolumn{2}{c}{Gemini} \\
\cmidrule(lr){2-3}\cmidrule(lr){4-5}
Method &
Exposure$\downarrow$ & Utility$\uparrow$ &
Exposure$\downarrow$ & Utility$\uparrow$ \\
\midrule
AbI Full
& 0.2134 & \best{0.914}
& 0.2229 & 0.950 \\
Ours Full
& \best{0.0793} & 0.910
& \best{0.0831} & \best{0.966} \\
\bottomrule
\end{tabularx}
\PGTableNote{Exposure and utility use PrivScreen privacy and visual questions. Bold marks the better endpoint between the Full pipelines. Privacy and utility remain separate outcomes.}
\end{table}

\subsubsection{AndroidLab Full Task Utility}
\label{subsec:eval-android}

AndroidLab evaluates task success on the same 138 tasks for four planners. A retained-output audit confirms that both protection pipelines repeatedly process private-content surfaces in this workload. Appendix Table~\ref{tab:eval-android-protection-audit} reports the review protocol and complete counts. PrivacyGuard's point estimate exceeds AbI for Gemini-3-Flash-Preview, Qwen3.7-Max, and GPT-4o, and ties AbI for Qwen3.7-Plus. The paired PrivacyGuard-minus-AbI differences are $+5.80$ [$+0.72,+10.87$], $0.00$ [$-7.25,+7.25$], $+0.72$ [$-5.07,+6.52$], and $+3.62$ [$-4.35,+11.59$] percentage points in the table order. Only the Gemini-3-Flash-Preview interval excludes zero. Relative to unprotected execution, PrivacyGuard Full differs by $-3.62$ to $+2.17$ points.

\begin{table}[t]
\centering
\caption{AndroidLab task success on the same 138 tasks.}
\label{tab:eval-android-main}
\begin{tabularx}{\linewidth}{@{}L{1.95}C{0.6}C{0.7}C{0.75}@{}}
\toprule
Planner & Original & AbI Full & Ours Full \\
\midrule
Gemini-3-Flash-Preview
& 60.87 & 57.25 & \best{63.04} \\
Qwen3.7-Plus
& 55.80 & \best{56.52} & \best{56.52} \\
Qwen3.7-Max
& 58.70 & 55.07 & \best{55.80} \\
GPT-4o
& 45.65 & 38.41 & \best{42.03} \\
\bottomrule
\end{tabularx}
\PGTableNote{Values are task-success percentages. Bold compares the two protection pipelines within each planner. Appendix Table~\ref{tab:appendix-android-complete} retains the Ours SAN-only rows used for RQ2.}
\end{table}

RQ3 shows lower measured upload-view exposure for PrivacyGuard Full and equal or higher task-success point estimates than AbI Full under the four evaluated planners, although only one paired task-success interval excludes zero.

\subsection{RQ4: Packaged-Prototype Feasibility and Runtime Cost}
\label{subsec:eval-runtime}

All 16 Xiaomi 12S runs completed without a process-level failure, and 14 completed the GUI task. The runs produced 182 sanitize calls and 87 restore calls. Mean upload-stage latency was 1378.4\,ms, of which OCR accounted for 1142.6\,ms; entity binding averaged 18.6\,ms, and return resolution and restoration averaged 15.5\,ms. In the separate server profile, PrivacyGuard had the lowest peak RSS among the protected pipelines at 2468.0\,MiB, compared with 2671.5--3551.5\,MiB, and its checkpoint-free detector had a 56.3\,ms p50.

\begin{table}[t]
\centering
\caption{Packaged Android-prototype latency and memory.}
\label{tab:eval-runtime-cost}
\begin{tabularx}{\linewidth}{@{}L{1.3}R{0.7}@{}}
\toprule
Metric & Value \\
\midrule
OCR latency, mean & 1142.6 ms \\
Preprocessing latency, mean & 60.5 ms \\
Detection latency, mean & 35.7 ms \\
Entity-binding latency, mean & 18.6 ms \\
Reference-rendering latency, mean & 120.9 ms \\
Return resolution and restore, mean & 15.5 ms \\
Return resolution and restore, median & 6.2 ms \\
Return resolution and restore, p95 & 47.3 ms \\
Ready-state USS, mean & 421.9 MiB \\
Working-state USS, mean / p95 & 969.0 / 1042.4 MiB \\

\bottomrule
\end{tabularx}
\PGTableNote{USS is private dirty plus private clean memory from \texttt{dumpsys meminfo}. Ready state averages ten temporal samples after one startup; working state summarizes 182 post-sanitize snapshots from the persistent process. Neither is an in-call peak.}
\end{table}

RQ4 establishes packaged execution on the evaluated handset, with OCR dominating upload latency.


\section{Related Work}
\label{sec:related-work}

GUI agents ground instructions in screen observations and execute multi-step actions on mobile, desktop and web interfaces~\cite{zhang2025appagent,li2024appagent,jiang2025appagentx,wang2024mobile,wang2024mobilev2,ye2025mobilev3,xu2026mobilev35,qin2025uitars}. Prior work notes trustworthiness and privacy risks in GUI interaction~\cite{shi2025towards,chen2025obvious}. In cloud-backed agents, instructions and UI state may be transmitted to remote models, making the device-to-cloud release point a privacy boundary~\cite{fan2026core,shen2026caped,wang2026guiguard,zhao2026anonymization}. PrivacyGuard mediates this boundary along the prompt--screenshot--return path.

GUI-specific defenses reduce or transform cloud-visible observations. CORE selects task-relevant UI blocks, MMPro exposes abstract representations to cloud reasoning while keeping perception and execution local, and PrivWeb combines local anonymization with user control~\cite{fan2026core,wu2025mmpro,zhang2026privweb}. DualTAP perturbs PII regions, whereas CAPED masks task-incidental content while retaining required evidence~\cite{zhang2025dualtap,shen2026caped}. GUIGuard-Bench evaluates privacy recognition, protection, and protected execution~\cite{wang2026guiguard}. These works govern which evidence is exposed or transformed; PrivacyGuard protects task-relevant PII that must remain referable without releasing its raw value.

Text-centric privacy work spans entity hiding and restoration~\cite{chen2023has,zeng2025privacyrestore}, differentially private prompt or embedding perturbation~\cite{tong2025inferdpt,mai2023splitdenoise}, private prompt learning~\cite{duan2023flocks}, comparative evaluation of protection mechanisms~\cite{agyemang2026llmredactor}, and type-consistent surrogate substitution~\cite{jathanna2026surrogateshield}. Presidio and GLiNER exemplify rule- and context-aware PII detection and open-label named-entity recognition, while PIIBench and WebPII benchmark PII detection in text and web screenshots, respectively~\cite{MsPresidio,zaratiana2024gliner,jha2026piibench,zhao2026webpii}. PrivacyGuard carries accepted occurrences beyond per-input detection and protection into source-linked rewriting and a session-scoped reference lifecycle.

Adjacent structural work addresses different tasks. ScreenAI, OmniParser, and SeeClick use learned vision or detection models to annotate or ground UI elements, while DiffSTR uses diffusion-based scene-text erasure~\cite{baechler2024screenai,lu2024omniparser,cheng2024seeclick,pathak2024diffstr}. Address systems respectively use prompt-driven LLM parsing and OCR-based BM25 retrieval with BERT reranking to structure or match records~\cite{tarannum2026system,gupta2021improvement}. PrivacyGuard instead applies deterministic geometry and local type-specific rules over OCR output to reconstruct source extents, bind address components, and resolve projections.

Among prior systems, AbI is the closest prior full-pipeline value-anonymization system~\cite{zhao2026anonymization}. It combines GLiNER-family detection with regex validation, replaces detected PII with deterministic type-preserving placeholders, and maintains a session-level value map with exact lookup and fuzzy OCR alignment. A local proxy resolves placeholders in agent commands before execution. PrivacyGuard's SAN instead reconstructs fragmented OCR flow, adjudicates partial or overlapping hypotheses, and emits source-linked extents. REF binds non-identical partial forms to entity roots and supports locally resolved component projections, providing a component-bearing lifecycle beyond AbI's value-level map.

\section{Limitations}
\label{sec:limitations}

\textbf{Empirical scope.}
The controlled corpus provides exact gold labels, while the structured-address split varies forms over the same source personas. The structured-address suite stress-tests coverage over declared transformation families and does not estimate their prevalence in natural GUI-agent traffic. PrivScreen provides external screenshots and value-independent questions. AndroidLab contributes public GUI tasks and retained protection outputs along complete trajectories. The generated ambiguity and fixed live-interface suites stress binding and reference-required execution. The source-value leakage audit covers two hosted readers from distinct model families, exact or partial recovery from released upload views, and targeted exact reconstruction from released type, root equality, projection, layout, and optional length cues. Gemini is identified by the recorded OpenRouter model string, request tuple, and cached responses. AndroidLab intervals resample tasks over fixed planner outputs. Repeated hosted-model runs remain a separate source of variation.

\textbf{Deployment and security scope.}
Compared protection pipelines run on the same RTX 4090, and the Xiaomi 12S workload establishes single-device packaged-prototype feasibility. Checkpoint-free denotes the PII detector. The screenshot path uses learned OCR. In normal use, evaluation measures confidentiality on individual released views and ordered five-turn Full upload histories. PrivacyGuard does not address planners that deliberately vary returned actions to infer local session state from later GUI feedback. Action authorization and execution integrity are separate GUI-agent properties. The complementary SAN-only and Full results motivate future task-aware reference release while preserving the same source-sanitization boundary. The current projection contract is address-centric. Extending it to other PII types and evaluating more applications and devices remain future work.

\section{Conclusion}
\label{sec:conclusion}

PrivacyGuard addresses the private reference problem by combining source-linked sanitization with a structured reference lifecycle. Source traces let SAN rewrite complete source extents, while typed entity records let REF maintain cross-form continuity, resolve supported projections locally, and restore values without returning them to the upload path. On common-gold five-turn uploads, PrivacyGuard lowers audited source-value leakage by 33.55\%--82.60\% relative to AbI Full in all reader--language cells. External PrivScreen results show lower final-view exposure, while REF raises strict cross-form continuity to about 0.81 and improves verified reference-required completion over exact-value reuse. AndroidLab shows equal or higher task-success point estimates than AbI Full for all four planners, with one paired interval excluding zero.

\clearpage
\appendix
\section*{Ethical Considerations}
\label{sec:ethical-considerations}

This study recruits no participants, performs no human-subject intervention, and uses no captured personal conversations. Controlled personas provide the exact values, source regions, within-session identities, and projections required for scoring without collecting or retaining participant or third-party PII. Persona-like values are nevertheless handled as benchmark-sensitive data. Public benchmarks and Android environments are used under their source access conditions, and hosted endpoints receive benchmark inputs rather than newly collected participant data.

The evaluated risks include raw-content disclosure after missed protection, task-context loss after over-detection, and action on the wrong protected entity after ambiguous binding. Type, equality, projection, length, and layout metadata can also support contextual inference. Raw mappings and resolution remain in the trusted mediation tier. Released files are screened for credentials and machine-local identifying paths, while third-party data remains under its original license and access terms.

\paragraph{Generative-AI disclosure.}
Author-directed generative-model assistance was used for controlled scenario and dialogue wording, followed by curation. Deterministic scripts instantiated gold values, spans, and entity links and produced pass/fail labels and measurements. GPT-5.6-Sol assisted the disclosed three-way AndroidLab detector-output audit against retained screen and accessibility-tree evidence. One author manually reviewed all 557 final deduplicated surfaces. The release includes the audit-form generator, adjudication schema, and aggregation code; screen-level evidence and the completed ledger are not redistributed. This audit reports method-conditioned output composition. OpenAI Codex (GPT-5) assisted manuscript restructuring and language revision. The authors verified the technical claims, citations, tables, and conclusions.

\section*{Open Science}
\label{sec:open-science}

\noindent\textbf{Anonymous source:} \url{https://github.com/betterzian/PrivacyGuard}. \textbf{Anonymous data and results:} \url{https://osf.io/7jsvq/overview?view_only=d71631e402ea4d36b547d18e5c4dd341}. These materials provide the implementation, fixed configurations, retained evaluation inputs, and analysis code needed to audit the reported results. Appendix~\ref{sec:appendix-provenance} gives the exact inventory, provenance, and redistribution boundaries.

\bibliographystyle{plainurl}
\bibliography{references}

@article{baechler2024screenai,
  title={Screenai: A vision-language model for ui and infographics understanding},
  author={Baechler, Gilles and Sunkara, Srinivas and Wang, Maria and Zubach, Fedir and Mansoor, Hassan and Etter, Vincent and C{\u{a}}rbune, Victor and Lin, Jason and Chen, Jindong and Sharma, Abhanshu},
  journal={arXiv preprint arXiv:2402.04615},
  year={2024}
}

@article{chen2023has,
  title={Hide and seek (has): A lightweight framework for prompt privacy protection},
  author={Chen, Yu and Li, Tingxin and Liu, Huiming and Yu, Yang},
  journal={arXiv preprint arXiv:2309.03057},
  year={2023}
}

@inproceedings{cheng2024seeclick,
  title={Seeclick: Harnessing gui grounding for advanced visual gui agents},
  author={Cheng, Kanzhi and Sun, Qiushi and Chu, Yougang and Xu, Fangzhi and YanTao, Li and Zhang, Jianbing and Wu, Zhiyong},
  booktitle={Proceedings of the 62nd Annual Meeting of the Association for Computational Linguistics (Volume 1: Long Papers)},
  pages={9313--9332},
  year={2024}
}

@article{duan2023flocks,
  title={Flocks of stochastic parrots: Differentially private prompt learning for large language models},
  author={Duan, Haonan and Dziedzic, Adam and Papernot, Nicolas and Boenisch, Franziska},
  journal={Advances in Neural Information Processing Systems},
  volume={36},
  pages={76852--76871},
  year={2023}
}

@article{mai2023splitdenoise,
  title={Split-and-denoise: Protect large language model inference with local differential privacy},
  author={Mai, Peihua and Yan, Ran and Huang, Zhe and Yang, Youjia and Pang, Yan},
  journal={arXiv preprint arXiv:2310.09130},
  year={2023}
}

@article{fan2026core,
  title={CORE: Reducing UI exposure in mobile agents via collaboration between cloud and local LLMs},
  author={Fan, Gucongcong and Niu, Chaoyue and Lyu, Chengfei and Wu, Fan and Chen, Guihai},
  journal={Advances in Neural Information Processing Systems},
  volume={38},
  pages={68690--68726},
  year={2026}
}

@article{jha2026piibench,
  title={PIIBench: A unified multi-source benchmark corpus for personally identifiable information detection},
  author={Jha, Pritesh},
  journal={arXiv preprint arXiv:2604.15776},
  year={2026}
}

@article{lu2024omniparser,
  title={Omniparser for pure vision based gui agent},
  author={Lu, Yadong and Yang, Jianwei and Shen, Yelong and Awadallah, Ahmed},
  journal={arXiv preprint arXiv:2408.00203},
  year={2024}
}

@article{shi2025towards,
  title={Towards trustworthy gui agents: A survey},
  author={Shi, Yucheng and Yu, Wenhao and Huang, Jingyuan and Yao, Wenlin and Chen, Wenhu and Liu, Ninghao},
  journal={arXiv preprint arXiv:2503.23434},
  year={2025}
}

@misc{chen2025obvious,
  title={The obvious invisible threat: Llm-powered gui agents’ vulnerability to fine-print injections},
  author={Chen, Chaoran and Zhang, Zhiping and Guo, Bingcan and Ma, Shang and Khalilov, Ibrahim and Gebreegziabher, Simret and Ye, Yanfang and Xiao, Ziang and Yao, Yaxing and Li, Tianshi and others},
  booktitle={Soups},
  year={2025},
  organization={The Twenty-First Symposium on Usable Privacy and Security (SOUPS)}
}

@misc{MsPresidio,
  AUTHOR={Mendels, Omri and Peled, Coby and Vaisman Levy, Nava and Hart, Sharon and Rosenthal, Tomer and Lahiani, Limor and others},
  TITLE={{Microsoft Presidio}: Context aware, pluggable and customizable PII anonymization service for text and images},
  YEAR={2018},
  ORGANIZATION={Microsoft},
  URL={https://microsoft.github.io/presidio},
  NOTE={Accessed: 2026-05-12}
}

@article{pathak2024diffstr,
  title={Diffstr: Controlled diffusion models for scene text removal},
  author={Pathak, Sanhita and Kaushik, Vinay and Lall, Brejesh},
  journal={arXiv preprint arXiv:2410.21721},
  year={2024}
}

@article{qin2025uitars,
  title={Ui-tars: Pioneering automated gui interaction with native agents},
  author={Qin, Yujia and Ye, Yining and Fang, Junjie and Wang, Haoming and Liang, Shihao and Tian, Shizuo and Zhang, Junda and Li, Jiahao and Li, Yunxin and Huang, Shijue and others},
  journal={arXiv preprint arXiv:2501.12326},
  year={2025}
}

@article{tarannum2026system,
  title={A System for Name and Address Parsing with Large Language Models},
  author={Tarannum, Adeeba and Mohammed, Muzakkiruddin Ahmed and Cakmak, Mert Can and Mandalawi, Shames Al and Talburt, John},
  journal={arXiv preprint arXiv:2601.18014},
  year={2026}
}

@inproceedings{gupta2021improvement,
  title={Improvement in semantic address matching using natural language processing},
  author={Gupta, Vansh and Gupta, Mohit and Garg, Jai and Garg, Nitesh},
  booktitle={2021 2nd International Conference for Emerging Technology (INCET)},
  pages={1--5},
  year={2021},
  organization={IEEE}
}

@article{tong2025inferdpt,
  title={Inferdpt: Privacy-preserving inference for closed-box large language models},
  author={Tong, Meng and Chen, Kejiang and Zhang, Jie and Qi, Yuang and Zhang, Weiming and Yu, Nenghai and Zhang, Tianwei and Zhang, Zhikun},
  journal={IEEE Transactions on Dependable and Secure Computing},
  volume={22},
  number={5},
  pages={4625--4640},
  year={2025},
  publisher={IEEE}
}

@article{wang2024mobile,
  title={Mobile-agent: Autonomous multi-modal mobile device agent with visual perception},
  author={Wang, Junyang and Xu, Haiyang and Ye, Jiabo and Yan, Ming and Shen, Weizhou and Zhang, Ji and Huang, Fei and Sang, Jitao},
  journal={arXiv preprint arXiv:2401.16158},
  year={2024}
}

@article{wang2024mobilev2,
  title={Mobile-agent-v2: Mobile device operation assistant with effective navigation via multi-agent collaboration},
  author={Wang, Junyang and Xu, Haiyang and Jia, Haitao and Zhang, Xi and Yan, Ming and Shen, Weizhou and Zhang, Ji and Huang, Fei and Sang, Jitao},
  journal={Advances in Neural Information Processing Systems},
  volume={37},
  pages={2686--2710},
  year={2024}
}

@article{ye2025mobilev3,
  title={Mobile-agent-v3: Fundamental agents for gui automation},
  author={Ye, Jiabo and Zhang, Xi and Xu, Haiyang and Liu, Haowei and Wang, Junyang and Zhu, Zhaoqing and Zheng, Ziwei and Gao, Feiyu and Cao, Junjie and Lu, Zhengxi and others},
  journal={arXiv preprint arXiv:2508.15144},
  year={2025}
}

@article{xu2026mobilev35,
  title={Mobile-agent-v3. 5: Multi-platform fundamental gui agents},
  author={Xu, Haiyang and Zhang, Xi and Liu, Haowei and Wang, Junyang and Zhu, Zhaozai and Zhou, Shengjie and Hu, Xuhao and Gao, Feiyu and Cao, Junjie and Wang, Zihua and others},
  journal={arXiv preprint arXiv:2602.16855},
  year={2026}
}

@article{wang2026guiguard,
  title={Guiguard: Toward a general framework for privacy-preserving gui agents},
  author={Wang, Yanxi and Zhang, Zhiling and Zhou, Wenbo and Zhang, Weiming and Zhang, Jie and Zhu, Qiannan and Shi, Yu and Zheng, Shuxin and He, Jiyan},
  journal={arXiv e-prints},
  pages={arXiv--2601},
  year={2026}
}

@article{zhao2026webpii,
  title={WebPII: Benchmarking Visual PII Detection for Computer-Use Agents},
  author={Zhao, Nathan},
  journal={arXiv preprint arXiv:2603.17357},
  year={2026}
}

@inproceedings{zaratiana2024gliner,
  title={GLiNER: Generalist model for named entity recognition using bidirectional transformer},
  author={Zaratiana, Urchade and Tomeh, Nadi and Holat, Pierre and Charnois, Thierry},
  booktitle={Proceedings of the 2024 conference of the north american chapter of the association for computational linguistics: Human language technologies (volume 1: Long papers)},
  pages={5364--5376},
  year={2024}
}

@inproceedings{zeng2025privacyrestore,
  title={Privacyrestore: Privacy-preserving inference in large language models via privacy removal and restoration},
  author={Zeng, Ziqian and Wang, Jianwei and Yang, Junyao and Lu, Zhengdong and Li, Haoran and Zhuang, Huiping and Chen, Cen},
  booktitle={Proceedings of the 63rd Annual Meeting of the Association for Computational Linguistics (Volume 1: Long Papers)},
  pages={10821--10855},
  year={2025}
}

@inproceedings{zhang2025appagent,
  title={Appagent: Multimodal agents as smartphone users},
  author={Zhang, Chi and Yang, Zhao and Liu, Jiaxuan and Li, Yanda and Han, Yucheng and Chen, Xin and Huang, Zebiao and Fu, Bin and Yu, Gang},
  booktitle={Proceedings of the 2025 CHI conference on human factors in computing systems},
  pages={1--20},
  year={2025}
}

@article{li2024appagent,
  title={Appagent v2: Advanced agent for flexible mobile interactions},
  author={Li, Yanda and Zhang, Chi and Jiang, Wenjia and Yang, Wanqi and Fu, Bin and Cheng, Pei and Chen, Xin and Chen, Ling and Wei, Yunchao},
  journal={arXiv preprint arXiv:2408.11824},
  year={2024}
}

@article{jiang2025appagentx,
  title={Appagentx: Evolving gui agents as proficient smartphone users},
  author={Jiang, Wenjia and Zhuang, Yangyang and Song, Chenxi and Yang, Xu and Zhou, Joey Tianyi and Zhang, Chi},
  journal={arXiv preprint arXiv:2503.02268},
  year={2025}
}

@article{zhang2025dualtap,
  title={DualTAP: A Dual-Task Adversarial Protector for Mobile MLLM Agents},
  author={Zhang, Fuyao and Zhang, Jiaming and Wang, Che and Sun, Xiongtao and Hao, Yurong and Guan, Guowei and Li, Wenjie and Huang, Longtao and Lim, Wei Yang Bryan},
  journal={arXiv preprint arXiv:2511.13248},
  year={2025}
}

@inproceedings{zhang2026privweb,
  title={Privweb: unobtrusive and content-aware privacy protection for web agents},
  author={Zhang, Shuning and Jiang, Yutong and Ma, Rongjun and Yang, Yuting and Xu, Mingyao and Huang, Zhixin and Yi, Xin and Li, Hewu},
  booktitle={Proceedings of the 2026 CHI Conference on Human Factors in Computing Systems},
  pages={1--29},
  year={2026}
}

@article{zhao2026anonymization,
  title={Anonymization-Enhanced Privacy Protection for Mobile GUI Agents: Available but Invisible},
  author={Zhao, Lepeng and Zou, Zhenhua and Li, Shuo and Liu, Zhuotao},
  journal={arXiv preprint arXiv:2602.10139},
  year={2026}
}

@inproceedings{wu2025mmpro,
  title={MMPro: A Decoupled Perception-Thinking-Execution Framework for Secure GUI Agent},
  author={Wu, Benlong and Qi, Yuang and Shang, Xiuwei and Zhang, Weiming and Yu, Nenghai and Chen, Kejiang},
  booktitle={Proceedings of the 33rd ACM International Conference on Multimedia},
  pages={4679--4687},
  year={2025}
}

@inproceedings{xu-etal-2025-androidlab,
  title={Androidlab: Training and systematic benchmarking of android autonomous agents},
  author={Xu, Yifan and Liu, Xiao and Sun, Xueqiao and Cheng, Siyi and Yu, Hao and Lai, Hanyu and Zhang, Shudan and Zhang, Dan and Tang, Jie and Dong, Yuxiao},
  booktitle={Proceedings of the 63rd Annual Meeting of the Association for Computational Linguistics (Volume 1: Long Papers)},
  pages={2144--2166},
  year={2025}
}

@article{shen2026caped,
  title={CAPED: Context-Aware Privacy Exposure Defense for Mobile GUI Agents},
  author={Shen, Siyu and Xu, Fenghao and Diao, Wenrui and Zhang, Kehuan},
  journal={arXiv preprint arXiv:2606.12666},
  year={2026}
}

@article{zhang2024large,
  title={Large language model-brained gui agents: A survey},
  author={Zhang, Chaoyun and He, Shilin and Qian, Jiaxu and Li, Bowen and Li, Liqun and Qin, Si and Kang, Yu and Ma, Minghua and Liu, Guyue and Lin, Qingwei and others},
  journal={arXiv preprint arXiv:2411.18279},
  year={2024}
}

@article{agyemang2026llmredactor,
  title={LLM-Redactor: An Empirical Evaluation of Eight Techniques for Privacy-Preserving LLM Requests},
  author={Agyemang, Justice Owusu and Kponyo, Jerry John and Amponsah, Elliot and Boakye, Godfred Manu Addo and Agyekum, Kwame Opuni-Boachie Obour},
  journal={arXiv preprint arXiv:2604.12064},
  year={2026}
}

@article{jathanna2026surrogateshield,
  title={SurrogateShield: Beyond Redaction for High-Utility, Privacy-Preserving LLM Interactions},
  author={Jathanna, Sherwin Vishesh},
  journal={arXiv preprint arXiv:2606.29567},
  year={2026}
}

\clearpage

\section{Algorithmic Details}
\label{sec:appendix-algorithms}

The optional local privacy repository supplies persistent device-local detection clues; it never populates or authorizes the session mapping used by REF. Appendix~\ref{sec:appendix-store-separation} specifies its schema, update semantics, and separation from session state. All headline comparisons use an empty repository; the matched pre/post diagnostic appears in Table~\ref{tab:eval-san-mechanisms-main}.

Algorithms~\ref{alg:entity-binding} and~\ref{alg:local-reference-resolution} state the two session-sensitive transitions summarized in Section~\ref{sec:method}. SAN candidate construction follows the source-linked pipeline described there: reconstruct source streams, emit typed clues with source traces, adjudicate them jointly, and retain only parser-supported candidates.

\section{Detailed Experimental Results}

This section collects the complete prompt, structured-address, reference-lifecycle, screenshot, reader-audit, and restoration results behind the compact main-text comparisons. Grouping these related results keeps each subsection's explanatory text with its tables while allowing page-wide floats to share the available appendix pages.

\subsection{Detailed Prompt Results}
\label{sec:appendix-prompt}

Table~\ref{tab:appendix-prompt-structured} expands the structured prompt results behind Table~\ref{tab:eval-text-address-main}, and Table~\ref{tab:appendix-prompt-per-type} reports the complete type-level cells summarized in RQ1. The separately distributed text benchmark also provides the matched perturbed split, which changes surface forms while preserving protected entities and source spans. Section~\ref{sec:appendix-androidlab-audit} gives the AndroidLab detector-output review procedure and uncertainty results.

\begin{algorithm}[htp]
\caption{Session-aware entity binding}
\label{alg:entity-binding}
\begin{algorithmic}[1]
    \Require SAN-accepted candidates $C_t$, prior session state $H_{1:t-1}$
    \Ensure Entity assignments $\Gamma_t$, updated state $H_{1:t}$
    \State $H \gets \operatorname{Copy}(H_{1:t-1})$; $\Gamma_t \gets \emptyset$
    \State $G_t \gets \textsc{AssembleCurrentTurnAppearances}(C_t)$
    \ForAll{$c \in \textsc{OrderForBinding}(C_t,G_t)$}
        \If{$c$ belongs to a group assigned earlier in turn $t$}
            \State $e \gets$ that group's assigned entity
        \Else
            \State $Q \gets \textsc{LookupIndexedRoots}(c,H)$
            \State $e \gets \textsc{FirstTypeCompatible}(Q,c)$
            \If{$e=\bot$}
                \State $e \gets \textsc{AllocateFreshEntity}(c)$
            \EndIf
        \EndIf
        \State $e \gets \textsc{MergeCompatibleState}(e,c)$
        \State Record $c$ under $e$; update indexes and recency in $H$
        \State $\Gamma_t[c] \gets e$
    \EndFor
    \State \Return $\Gamma_t,H$
\end{algorithmic}
\end{algorithm}

\begin{algorithm}[htp]
\caption{Local reference resolution and restoration}
\label{alg:local-reference-resolution}
\begin{algorithmic}[1]
    \Require Returned action $A'_t$, active session state $H_{1:t}$
    \Ensure Locally resolved action payload $A_t$
    \State $A_t \gets A'_t$
    \ForAll{complete bracketed candidate $z$ in $A'_t$}
        \State $\bar z \gets \textsc{NormalizeAllowlisted}(z)$
        \State $(\tau,n,\sigma,\lambda) \gets \textsc{ParseAnchored}(\bar z)$
        \Comment{$\sigma$: address projection; $\lambda$: optional \texttt{LEN=L}}
        \If{parsing fails}
            \State \textbf{continue}
        \EndIf
        \State $e \gets \textsc{LookupActiveRoot}(H_{1:t},n)$
        \If{$e=\bot$ or $\operatorname{type}(e)\neq\tau$}
            \State \textbf{continue}
        \EndIf
        \If{$\sigma=\varnothing$ and $\lambda=\varnothing$}
            \State $v \gets \textsc{AccumulatedValue}(e)$
        \ElsIf{$\sigma=\varnothing$, $\tau\in\{\mathrm{NUM},\mathrm{ALNUM}\}$, and $\lambda=\texttt{LEN=L}$ for an integer $L\geq0$}
            \State $v \gets \textsc{AccumulatedValue}(e)$
            \Comment{descriptive only; do not compare $L$ with the stored value}
        \ElsIf{$\sigma$ is supported and locally present in $e$}
            \State $v \gets \textsc{Project}(e,\sigma)$
        \Else
            \State \textbf{continue}
        \EndIf
        \State Replace this occurrence of $z$ by $v$ in $A_t$
    \EndFor
    \State \Return $A_t$
\end{algorithmic}
\end{algorithm}

\begin{table*}[!t]
\centering
\caption{Full structured-prompt results.}
\label{tab:appendix-prompt-structured}
\begin{tabularx}{\linewidth}{@{}L{1.5}C{0.4}C{1.02}C{1.02}C{1.02}C{1.02}C{1.02}@{}}
\toprule
Method & Lang & GenR@1.0$\uparrow$ & GenR@0.8$\uparrow$ & $\mathrm{LR}_{\mathrm{char}}\downarrow$ & EPred\_gt$\downarrow$ & ECov\_char$\downarrow$ \\
\midrule
Presidio + rules & EN & 0.4985 & 0.4985 & 0.5940 & 0.1250 & 0.5356 \\
PrA & EN & 0.6749 & 0.6749 & 0.2627 & 0.4782 & 0.0070 \\
PrG & EN & 0.7912 & 0.8593 & 0.1283 & 0.3089 & 0.0829 \\
AbI & EN & 0.5456 & 0.8212 & 0.2979 & 0.7252 & \best{0.0000} \\
Ours & EN & \best{0.9740} & \best{0.9794} & \best{0.0125} & \best{0.0016} & \best{0.0000} \\
\midrule
Presidio + rules & ZH & 0.4725 & 0.4725 & 0.4767 & \best{0.0000} & \best{0.0000} \\
PrA & ZH & 0.5048 & 0.5048 & 0.3998 & 0.2210 & 0.0022 \\
PrG & ZH & 0.3231 & 0.3231 & 0.6876 & \best{0.0000} & 0.0185 \\
AbI & ZH & 0.2709 & 0.2709 & 0.7256 & \best{0.0000} & \best{0.0000} \\
Ours & ZH & \best{0.9971} & \best{0.9971} & \best{0.0023} & \best{0.0000} & 0.0087 \\
\bottomrule
\end{tabularx}
\PGTableNote{\textit{Note.} The protocol matches Table~\ref{tab:eval-text-address-main}. Best values are marked within each language; all tied cells are marked.}
\end{table*}

\begin{table*}[!t]
\centering
\caption{Per-type structured-prompt source removal. Cells are GenR@1.0 / $\mathrm{LR}_{\mathrm{char}}$.}
\label{tab:appendix-prompt-per-type}
\begin{tabularx}{\linewidth}{@{}C{0.3}L{1.05}R{0.5}C{1.23}C{1.23}C{1.23}C{1.23}C{1.23}@{}}
\toprule
Lang & PII type & $N$ & AbI & PrG & Presidio + rules & PrA & Ours \\
\midrule
EN & Account ID & 126 & 0.9841/0.0202 & 0.6111/0.1759 & 0.9762/0.0151 & 0.0000/0.8652 & 0.9365/0.0469 \\
EN & Address & 625 & 0.0000/0.5208 & 0.0000/0.2245 & 0.0000/0.8237 & 0.0000/0.4218 & 0.9200/0.0136 \\
EN & Bank card & 150 & 1.0000/0.0000 & 1.0000/0.0000 & 1.0000/0.0000 & 1.0000/0.0000 & 1.0000/0.0000 \\
EN & Driver license & 328 & 1.0000/0.0000 & 1.0000/0.0000 & 1.0000/0.0000 & 0.9848/0.0000 & 1.0000/0.0000 \\
EN & EMAIL & 262 & 1.0000/0.0000 & 0.9924/0.0004 & 1.0000/0.0000 & 1.0000/0.0000 & 1.0000/0.0000 \\
EN & License plate & 135 & 0.9926/0.0000 & 1.0000/0.0000 & 0.3704/0.5825 & 0.0000/0.8653 & 1.0000/0.0000 \\
EN & Member ID & 155 & 0.4387/0.5641 & 0.4065/0.4090 & 1.0000/0.0000 & 0.0129/0.8112 & 0.8516/0.1392 \\
EN & NAME & 1141 & 0.1543/0.0868 & 0.9965/0.0012 & 0.0000/0.8435 & 0.9579/0.0000 & 0.9869/0.0079 \\
EN & PHONE & 775 & 1.0000/0.0000 & 1.0000/0.0000 & 1.0000/0.0000 & 0.8581/0.0000 & 1.0000/0.0000 \\
\midrule
ZH & Account ID & 117 & 0.0000/1.0000 & 0.0171/0.9831 & 0.7265/0.2446 & 0.0000/0.9038 & 0.9402/0.0500 \\
ZH & Address & 668 & 0.0000/1.0000 & 0.0344/0.9557 & 0.0000/1.0000 & 0.0000/0.6815 & 1.0000/0.0000 \\
ZH & Bank card & 150 & 0.0000/1.0000 & 0.0267/0.9732 & 1.0000/0.0000 & 0.0000/0.7855 & 1.0000/0.0000 \\
ZH & EMAIL & 254 & 0.9685/0.0308 & 0.9764/0.0215 & 1.0000/0.0000 & 1.0000/0.0000 & 1.0000/0.0000 \\
ZH & ID card & 329 & 0.0000/0.9605 & 0.1185/0.8815 & 1.0000/0.0000 & 1.0000/0.0000 & 1.0000/0.0000 \\
ZH & License plate & 115 & 0.0000/1.0000 & 0.0000/1.0000 & 0.0000/1.0000 & 0.0000/0.7913 & 1.0000/0.0000 \\
ZH & Member ID & 155 & 0.0000/1.0000 & 0.0000/0.9925 & 1.0000/0.0000 & 0.0000/0.9195 & 0.9742/0.0277 \\
ZH & NAME & 1183 & 0.0423/0.9577 & 0.1505/0.8495 & 0.0000/1.0000 & 0.4683/0.4294 & 1.0000/0.0000 \\
ZH & PHONE & 817 & 0.8935/0.1065 & 0.8935/0.1065 & 1.0000/0.0000 & 0.9486/0.0000 & 1.0000/0.0000 \\
\bottomrule
\end{tabularx}
\PGTableNote{\textit{Note.} Results use the empty local privacy repository. Table~\ref{tab:eval-san-mechanisms-main} reports the separate repository post-pass sensitivity diagnostic.}
\end{table*}

\subsubsection{AndroidLab Detector-Output Audit Prompt}
\label{sec:appendix-androidlab-audit}

The audit unit is one normalized detector surface with its retained crops, screen context, accessibility text, and source traces. GPT-5.6-Sol proposes one of three content classes---protected, neutral optional, or explicit non-private---plus extraction-quality flags. One author then reviewed all 557 semantically deduplicated, method-conditioned surfaces (319 PrivacyGuard and 238 AbI), checking every distinct task context and splitting context-dependent cases. Table~\ref{tab:eval-android-protection-audit} uses the author-final labels. The anonymous source contains the audit-form generator, adjudication schema, and aggregation code. Screen-level evidence and the completed adjudication ledger are not part of the frozen results archive.

\begin{table}[t]
\centering
\caption{Author-reviewed AndroidLab protection-output composition.}
\label{tab:eval-android-protection-audit}
\setlength{\tabcolsep}{2pt}
\begin{tabularx}{\linewidth}{@{}L{0.75}R{0.7}R{1.2}R{1.15}R{1.2}@{}}
\toprule
Method & Surfaces & Protected & Neutral & Explicit neg. \\
\midrule
Ours & 319 & 269 (84.33\%) & 26 (8.15\%) & 24 (7.52\%) \\
AbI          & 238 & 182 (76.47\%) & 32 (13.45\%) & 24 (10.08\%) \\
\bottomrule
\end{tabularx}
\PGTableNote{Composition of emitted, method-conditioned surfaces after conservative semantic deduplication.}
\end{table}

Semantic deduplication merges only verified spelling, OCR, truncation, repetition, or formatting variants; substring containment alone never establishes identity.

\begin{table*}[!t]
\centering
\caption{Task-paired uncertainty for AndroidLab utility.}
\label{tab:eval-android-paired}
\begin{tabularx}{\linewidth}{@{}L{1.2}L{0.8}L{1.3}R{0.7}@{}}
\toprule
Planner & Baseline & $\Delta$SR pp [95\% CI] & Ours win/loss/tie \\
\midrule
Qwen3.7-Plus & SAN-only     & $6.52$ [$-1.45$, $14.49$] & 20/11/107 \\
Qwen3.7-Plus & Original & $0.72$ [$-5.80$, $7.25$] & 12/11/115 \\
Qwen3.7-Plus & AbI         & $0.00$ [$-7.25$, $7.25$] & 12/12/114 \\
GPT-4o       & Original & $-3.62$ [$-10.14$, $2.90$] & 8/13/117 \\
GPT-4o       & AbI         & $3.62$ [$-4.35$, $11.59$] & 17/12/109 \\
Gemini-3-Flash-Preview & SAN-only     & $2.90$ [$-2.90$, $8.70$] & 11/7/120 \\
Gemini-3-Flash-Preview & Original & $2.17$ [$-3.62$, $7.97$] & 10/7/121 \\
Gemini-3-Flash-Preview & AbI         & $5.80$ [$0.72$, $10.87$] & 11/3/124 \\
Qwen3.7-Max  & Original & $-2.90$ [$-9.42$, $3.62$] & 8/12/118 \\
Qwen3.7-Max  & AbI         & $0.72$ [$-5.07$, $6.52$] & 9/8/121 \\
\bottomrule
\end{tabularx}
\PGTableNote{Each row compares the same 138 tasks. Percentile intervals use 10,000 task-paired bootstrap resamples with seed 1729 and quantify task-resampling uncertainty over the fixed run outputs.}
\end{table*}

\begin{table}[t]
\centering
\caption{AndroidLab SAN-only-to-Full task-success contrast.}
\label{tab:appendix-android-complete}
\begin{tabularx}{\linewidth}{@{}L{2}C{0.8}C{0.5}C{0.5}@{}}
\toprule
Planner & SAN-only & Full & $\Delta$ pp \\
\midrule
Gemini-3-Flash-Preview & 60.14 & 63.04 & $+2.90$ \\
Qwen3.7-Plus & 50.00 & 56.52 & $+6.52$ \\
\bottomrule
\end{tabularx}
\PGTableNote{Values are percentages on the same 138 tasks. SAN-only removes detected values without readable references. Full-system comparisons with Original and AbI remain in Table~\ref{tab:eval-android-main}.}
\end{table}

\subsection{Detailed Structured-Address Results}
\label{sec:appendix-address}

Table~\ref{tab:appendix-address-pcr} gives pair coverage and conditional consistency for the cross-system continuity rows summarized in Table~\ref{tab:eval-ref-capability-continuity}; exact-reuse wrapper rows serve as detector controls. Table~\ref{tab:eval-address-spec-main} gives the SPEC-construction outcomes, and Table~\ref{tab:eval-address-split-main} gives the paired split-block BD/Chain matrix. The retained prompt/address summary CSVs include the complete upload-side baseline rows; their headline endpoints are already reported in Table~\ref{tab:eval-text-address-main}.

\begin{table*}[!t]
\centering
\caption{Representative SPEC-aware address projections.}
\label{tab:eval-address-spec-main}
\begin{tabularx}{\linewidth}{@{}L{0.8}C{0.5}C{1.6}C{0.5}C{1.6}@{}}
\toprule
Variant & EN $N$ & EN (Exact / Under / Empty) & ZH $N$ & ZH (Exact / Under / Empty) \\
\midrule
to\_city & 360/360 & 1.000 / 0.000 / 0.000 & 360/360 & 1.000 / 0.000 / 0.000 \\
to\_district & 360/360 & 1.000 / 0.000 / 0.000 & 360/360 & 0.936 / 0.064 / 0.000 \\
to\_road & 339/360 & 1.000 / 0.000 / 0.000 & 360/360 & 0.758 / 0.242 / 0.000 \\
full & 360/360 & 0.000 / 1.000 / 0.000 & 360/360 & 0.000 / 1.000 / 0.000 \\
ocr\_noise & 360/360 & 0.000 / 1.000 / 0.000 & 360/360 & 0.000 / 1.000 / 0.000 \\
\bottomrule
\end{tabularx}
\PGTableNote{$N$ is evaluable/total. Under is a nonempty strict subset of the expected SPEC-token set; Empty is an evaluable root with an empty SPEC. Empty and overshoot counts are both zero. The 21 non-evaluable EN road cases are SAN detection misses and are excluded from REF projection rates.}
\end{table*}

\begin{table*}[!t]
\centering
\caption{Detection and reassembly of split-block addresses.}
\label{tab:eval-address-split-main}
\begin{tabularx}{\linewidth}{@{}L{1.4}C{0.9}C{0.9}C{0.9}C{0.9}@{}}
\toprule
\multirow{2}{*}{Method} & \multicolumn{2}{c}{2-block} & \multicolumn{2}{c}{3-block} \\
\cmidrule(lr){2-3}\cmidrule(lr){4-5}
 & EN & ZH & EN & ZH \\
\midrule
Presidio + rules + Exact & 0.442 / 0.000 & 0.000 / 0.000 & 0.383 / 0.000 & 0.000 / 0.000 \\
PrA+Exact & 0.979 / 0.000 & 0.681 / 0.000 & 0.881 / 0.000 & 0.582 / 0.000 \\
PrG+Exact & \best{1.000} / 0.000 & 0.475 / 0.000 & \best{1.000} / 0.000 & 0.482 / 0.000 \\
AbI & 0.971 / 0.000 & 0.367 / 0.000 & 0.708 / 0.000 & 0.250 / 0.000 \\
Ours & 0.992 / \best{0.925} & \best{0.952} / \best{0.667} & 0.944 / \best{0.817} & \best{0.915} / \best{0.538} \\
\bottomrule
\end{tabularx}
\PGTableNote{Each cell reports BD / Chain. Presidio rows use the same Exact wrapper as Table~\ref{tab:appendix-address-pcr}; best values are marked per subcolumn.}
\end{table*}

\begin{table*}[!t]
\centering
\caption{Full cross-type continuity matrix at $\tau=0.80$.}
\label{tab:appendix-address-pcr}
\begin{tabularx}{\linewidth}{@{}C{0.4}L{1.5}C{1.025}C{1.025}C{1.025}C{1.025}@{}}
\toprule
Lang & Method & $C_{\mathrm{pair}}\uparrow$ & $\mathrm{PCR}_{\mathrm{cond}}\uparrow$ & $\mathrm{PCR}_{\mathrm{strict}}\uparrow$ & Chain$\uparrow$ \\
\midrule
EN & GLiNER + AbI                 & 0.2538 & 0.5773 & 0.1465 & 0.0000 \\
EN & PrA+Exact                    & 0.0882 & 0.3921 & 0.0346 & 0.0000 \\
EN & PrG+Exact                    & \best{0.9150} & 0.0546 & 0.0499 & 0.0000 \\
EN & Ours                         & 0.8886 & \best{0.9098} & \best{0.8084} & \best{0.8708} \\
\midrule
ZH & GLiNER + AbI                 & 0.0047 & 0.0000 & 0.0000 & 0.0000 \\
ZH & PrA+Exact                    & 0.0313 & 0.2351 & 0.0074 & 0.0000 \\
ZH & PrG+Exact                    & \best{0.9746} & 0.0076 & 0.0075 & 0.0000 \\
ZH & Ours                         & 0.9448 & \best{0.8652} & \best{0.8175} & \best{0.6021} \\
\bottomrule
\end{tabularx}
\PGTableNote{\textit{Note.} Presidio detector rows use the evaluation-only Exact wrapper: identical complete detected surfaces reuse a session-scoped identifier, with no component or fuzzy equivalence. AbI uses its native value map. Table~\ref{tab:eval-ref-capability-continuity} reports every row's strict PCR and chain result; this companion additionally separates pair coverage from conditional consistency. Chain is the mean two-/three-block chain rate. The wrapper is applied to frozen recorded detections by the released companion script.}
\end{table*}

\subsection{REF Continuity and Policy Details}
\label{sec:appendix-reference}

Table~\ref{tab:eval-ref-capability-continuity} reports Full results by planner relative to the structural SAN-only endpoint together with the matched lifecycle ablation. This appendix gives the per-case counts and failure taxonomy for turn-local references, AbI-style exact-value reuse, and the Full structured lifecycle described in RQ2. The anonymous source contains the deterministic scoring exporter and aggregation code. Every counted success across both planners and all three executed mechanisms follows a submitted reference through local resolution and restoration.

\paragraph{Resolution and restoration contract.}
An address mention exposes only the display levels present in that source mention. PrivacyGuard orders them as \texttt{COUNTRY}, \texttt{PROV}, \texttt{CITY}, \texttt{DIST}, \texttt{ROAD}, and \texttt{DTL}. A road-only mention can therefore render as \texttt{[[ADDR\#2.ROAD]]}, a unit fragment as \texttt{[[ADDR\#2.DTL]]}, and a multi-level mention as a combined \texttt{SPEC}. Every projection retains the entity root. The bare root denotes the accumulated complete address and is also used when a source mention exposes no supported display level. NAME family, given-name, and title components are binding evidence only in the evaluated configuration; the external component-projection contract remains address-specific.

\paragraph{Descriptive length metadata.}
When local context classifies an otherwise generic fragment as NUM or ALNUM, the renderer can retain its implementation-normalized length in \texttt{LEN=L}. NUM length ignores common spaces and hyphens and a leading sign; ALNUM counts the displayed fragment string. The cue distinguishes shapes such as a six-digit verification code and a ten-character meeting identifier. The grammar accepts \texttt{LEN=L} only for NUM/ALNUM with a nonnegative integer $L$; it is descriptive rather than an authorization or integrity field. Resolution therefore checks the active session root and stored type but does not compare $L$ with the stored value: for example, both \texttt{[[ALNUM\#8.LEN=8]]} and a syntactically valid \texttt{[[ALNUM\#8.LEN=9]]} resolve to the same accumulated ALNUM value. Address \texttt{LEN} suffixes remain unsupported. Contextual-reconstruction audits retain the released \texttt{LEN} cue and include its contribution.

\paragraph{Evaluated reference interfaces.}
The evaluation fixes three interfaces: turn-local references, AbI-style exact-value reuse, and PrivacyGuard's Full structured lifecycle. Table~\ref{tab:eval-ref-capability-continuity} gives their planner-separated totals; the tables below provide the per-case counts.

Under Qwen, AbI-style exact-value reuse completes 8/12 root-like and 0/18 component-like runs, while Full completes 11/12 and 6/18. Under Gemini, the corresponding counts are 7/12 and 0/18 for AbI-style exact-value reuse, versus 8/12 and 12/18 for Full. Turn-local references complete 0/18 component-like runs under both planners and 1/12 Gemini root-like runs. Thus the cross-planner result isolates Full's component-projection gain beyond identical-value reuse. The anonymous source includes the generic interface notice and deterministic exporter; Tables~\ref{tab:appendix-reference-cases} and~\ref{tab:eval-reference-failures} report the planner-separated counts and outcomes.

\begin{table*}[!t]
\centering
\caption{Automatically verified reference paths by planner and case.}
\label{tab:appendix-reference-cases}
\begin{tabularx}{\textwidth}{@{}L{1.1}L{1.0}L{0.9}C{1.0}C{1.0}C{1.0}@{}}
\toprule
Planner & Reference case & Family & Turn-local & Exact reuse & Full lifecycle \\
\midrule
\multirow{5}{*}{Qwen}
& PHONE & Root-like & 0/6 & 5/6 & 6/6 \\
& FULL\_ADDRESS & Root-like & 0/6 & 3/6 & 5/6 \\
& CITY & Component-like & 0/6 & 0/6 & 3/6 \\
& DISTRICT & Component-like & 0/6 & 0/6 & 1/6 \\
& DTL & Component-like & 0/6 & 0/6 & 2/6 \\
\midrule
\multirow{5}{*}{Gemini}
& PHONE & Root-like & 0/6 & 5/6 & 4/6 \\
& FULL\_ADDRESS & Root-like & 1/6 & 2/6 & 4/6 \\
& CITY & Component-like & 0/6 & 0/6 & 5/6 \\
& DISTRICT & Component-like & 0/6 & 0/6 & 2/6 \\
& DTL & Component-like & 0/6 & 0/6 & 5/6 \\
\bottomrule
\end{tabularx}
\PGTableNote{\textit{Note.} Each planner--case row contains six executions scored by the deterministic path rule. A success requires a submitted reference with a trajectory-issued root and the correct projection, followed by correct local restoration; a correct plaintext guess without that path is not counted. Full structured mediation introduces the projections exercised by CITY, DISTRICT, and DTL.}
\end{table*}

\begin{table*}[t]
\centering
\caption{Planner-separated automatic outcomes for Full structured mediation.}
\label{tab:eval-reference-failures}
\begin{tabularx}{\linewidth}{@{}L{1.15}C{0.13}C{0.17}L{2.55}@{}}
\toprule
Outcome & Qwen & Gemini & Automatic definition \\
\midrule
Successful reference path & 17 & 20 & One issued target reference has the required projection and restores to gold. \\
Reference not submitted & 7 & 5 & The final answer does not submit a target-class reference. \\
Projection or granularity mismatch & 5 & 0 & The submitted address projection does not match the requested component or requested full-address granularity. \\
Wrong issued-root selection & 0 & 3 & The planner submits another valid root issued earlier in the trajectory. \\
Spurious parsed root & 1 & 1 & The submitted root was issued from a detector/parser candidate that does not denote the target entity. \\
Unresolved or malformed restore & 0 & 1 & Restoration fails or leaves a reference token unresolved. \\
\bottomrule
\end{tabularx}
\PGTableNote{\textit{Note.} Each planner contributes 30 runs: 12 root-like and 18 component-like. Failure categories follow their originating stage: reference non-use and wrong issued-root selection are planner outcomes; spurious parsed roots originate in detection/parsing; projection mismatch originates in component selection; and malformed references are local-resolution failures.}
\end{table*}

\paragraph{Planner and task protocol.}
The study evaluates Qwen and Gemini separately through the Mobile-Agent-v3.5 GUI-Owl runner in their general pretrained configurations. Ten WeChat templates cover two order identities and five required outputs: PHONE, FULL\_ADDRESS, CITY, DISTRICT, and DTL. The warm-up protocol writes synthetic phone and address records into separate live chats. Each measured task then navigates named contacts, reads PII-bearing messages distributed across those chats, and returns the protected value or requested component. For each planner, every template is executed three times under every lifecycle, yielding 30 executions per planner--lifecycle cell and 180 executions overall. The runner enforces a 25-step budget. Within each planner, the task definitions and runner configuration are held fixed across mechanisms; the released reference lifecycle is the experimental factor. Turn-local mediation issues independent references, AbI-style exact-value reuse preserves session roots through exact canonical reuse, and Full structured mediation additionally supplies component-compatible binding and supported address projections. These executions use the live Android WeChat interface. Free-form chat distributes a referent across message bubbles, turns, and OCR regions. Each scored case tracks an issued reference from the mediated instruction or screen to the planner's submitted action and its local resolution and restoration. Full is compared with SAN-only's structural zero endpoint for capability and with the three executed lifecycles for mechanism ablation. Screenshot source-level metrics test SAN direct leakage, targeted reader audits measure REF reconstruction, and the 138-task AndroidLab study measures end-to-end planner utility.

The runner retains the stock mobile-action interface and adds one generic PrivacyGuard notice shared by every PrivacyGuard task. The notice explains that bracketed roots and projections are valid string arguments and will be restored locally; it contains no task-specific private value, target reference, or worked task solution. References themselves enter the planner context through the sanitized task instruction and screen, plus prior action history. That history retains the planner's pre-restoration reference-bearing output rather than the locally restored payload. The path audit follows this fixed interface through the final executed answer.

\paragraph{Strict automatic scoring rule.}
The exporter reads the final pre-restore action, all roots issued earlier in the same trajectory, the requested gold component, restore counters, and the post-restore action. A run succeeds iff the final action contains exactly one target-class reference; its root was issued in the trajectory, and its syntax and projection satisfy the component contract; local restoration changes the action, reports at least one successful and zero failed restores, and leaves no unresolved reference token; and the normalized restored action contains the gold value or component. For CITY, DISTRICT, and DTL tasks, a bare root or an over-broad address therefore fails the projection predicate. Any literal gold already present in the pre-restore action also makes the run a failure. The deterministic rule supplies the reported labels, and its Qwen labels agree with the previously completed all-row manual check in 90/90 cases. The anonymous source includes the task definitions, command schemas, scoring exporter, and aggregation scripts. Raw device trajectories are not part of the frozen version~0.1.0 research-results archive.

\subsubsection{Recent-First Policy Characterization}
\label{sec:appendix-binding-risk}

Let $O_B=(H,c)$ denote the complete normalized observation consumed by the binder: accumulated entity components, candidate components, type-specific indexes, and recency order. If roots $e_1$ and $e_2$ both pass the compatibility predicate for an underspecified candidate $c$, reversing their insertion order reverses the selected root. Both lookup implementations agree for either order, directly verifying the specified recent-first behavior.

\begin{table*}[!t]
\centering
\caption{Exhaustive recent-first conformance over observable binder states.}
\label{tab:binding-policy-behavior}
\begin{tabularx}{\textwidth}{@{}L{2}C{0.6}C{0.8}C{1}C{0.8}C{0.8}@{}}
\toprule
Policy & Unique reuse & Multi-root selection & Order-dependent groups & Policy violations & Path agreement \\
\midrule
Component-projection recent-first & 119/119 & 290/290 & 16/16 & 0/443 & 443/443 \\
\bottomrule
\end{tabularx}
\PGTableNote{The enumeration covers zero to four roots, every compatible subset, and every insertion order. Path agreement compares reconstructed-record and runtime-index lookup; concrete cases also give 8/8 complementary reuses, 8/8 conflict separations, and 8/8 most-recent ambiguous attachments.}
\end{table*}

Recent-first consumes only $O_B$ and selects by recency; the intended referent is a gold evaluation label, not an input to the binder. The same binder observation and selected root can therefore be labeled correct or incorrect under different intended referents. Agreement rates must be measured on a declared target distribution rather than inferred from the policy alone.

\begin{table*}[!t]
\centering
\caption{Generator-conditioned exact-continuation success under two binding policies on the same 6,000 gold continuations.}
\label{tab:appendix-binding-policy-stress}
\begin{tabularx}{\textwidth}{@{}L{1.2}C{0.95}C{0.95}C{0.95}C{0.95}@{}}
\toprule
Policy & Unique success & Ambiguous success & Overall exact success & Exact failures \\
\midrule
Recent-first & 4,236/4,236 & 1,050/1,764 & \best{5,286/6,000 (0.881)} & 714 \\
Fresh-on-ambiguity & 4,236/4,236 & 0/1,764 & 4,236/6,000 (0.706) & 1,764 \\
\bottomrule
\end{tabularx}
\PGTableNote{Exact-continuation success requires the selected root to equal the generator-provided intended prior root. Fresh allocation and non-gold prior-root selection are both failures. These results are conditioned on the declared generator and intended-root distribution.}
\end{table*}

The exhaustive conformance evaluation gives 0/443 policy violations and 443/443 agreement between reconstructed-record and runtime-index lookup. The generator-conditioned comparison addresses a different question. Recent-first agrees with the generator-provided intended root in 5,286 trials, compared with 4,236 for fresh-on-ambiguity. It preserves 1,050 ambiguous continuations that fresh allocation discards, but it selects a non-gold compatible root in 714 trials. The result characterizes policy--intent agreement for this generator and does not establish a general false-merge guarantee.

\paragraph{Dataset-conditioned policy--intent audit.}
The unified model defines the policy--intent disagreement rate on an annotated dataset $D$ as $R_D=\Pr_D(A)\Pr_D(F\mid A)$ in Equation~\ref{eq:dataset-semantic-impact}. Its empirical estimate is the number of non-gold selections divided by scored continuation assignments. For the controlled-corpus audit, one event is one gold-aligned emitted root group in a view whose selected root was already issued. Multiple OCR boxes carrying the same root in that view form one event. Gold identity maps address CITY, DISTRICT, ROAD, BUILDING, UNIT, and related component annotations into their complete-address entity; type-compatible canonical protected values provide the corresponding scalar equality rule.

Across the 480 five-turn sessions, 10{,}288 emitted root groups align to gold and 4{,}352 select an existing root. Recent-first yields 0/4{,}352 binder policy--intent disagreements: EN 0/2{,}048 and ZH 0/2{,}304. At chain level, a scored chain is a session-scoped root with at least two gold-aligned occurrences, and an impure chain contains more than one non-equivalent gold entity. The audit yields 0/2{,}031 impure chains: EN 0/1{,}005 and ZH 0/1{,}026. These annotations provide selection and chain outcomes. Complete compatible-root-set labels additionally instantiate $N_A$, $N_F$, $q_D$, $\rho_D$, $R_D$, and chain impurity; the generated stress suite supplies those labels.

\paragraph{Randomized stress protocol.}
The independent generator uses seed 20260725 and 6{,}000 trials. It samples one to four compatible roots with a decreasing count distribution and samples the gold root's recency rank with decreasing probability. Compatible roots share a generated road component but carry distinct complete-address numbers; incompatible roots act as distractors and insertion order is permuted. The comparison executes recent-first and fresh-on-ambiguity on the same rows through the production compatibility predicate and both lookup paths. The source snapshot fixes the generator code, parameters, and seed; the frozen summary records analytical expectations and aggregate compatible-set and root-count strata. The reported rates characterize this declared stress distribution rather than a fitted population.

\paragraph{Privacy interpretation.}
A visible root/projection edge can reveal compatibility, equality, substring, type, or component information. A target distribution with compatible-set and intended-identity labels determines $q_D$, $\rho_D$, and $R_D$. Once the policy-generated raw targets and cloud-visible reference graph are fixed, changing only the gold policy--intent agreement labels leaves a fixed reader's input and reconstruction output unchanged; each released policy graph therefore receives its own leakage audit.

\begin{table}[!t]
\centering
\caption{Complete CrossTurnPII single-view reader exposure on common gold. SAN rows are placeholder-free erase-only views. Full rows are the two reference-releasing pipelines.}
\label{tab:appendix-crossturnpii-complete}
\begin{tabularx}{\linewidth}{@{}C{0.45}L{1.25}C{0.55}C{1.375}C{1.375}@{}}
\toprule
Lang & Method & View & Qwen$\downarrow$ & Gemini$\downarrow$ \\
\midrule
EN & AbI & SAN & 0.2775 & 0.2800 \\
EN & PrA & SAN & 0.4758 & 0.4826 \\
EN & PrG & SAN & 0.5754 & 0.5815 \\
EN & Ours & SAN & \best{0.1289} & \best{0.1332} \\
\cmidrule(lr){2-5}
EN & AbI & Full & 0.3559 & 0.3444 \\
EN & Ours & Full & \best{0.1310} & \best{0.1353} \\
\midrule
ZH & AbI & SAN & 0.7513 & 0.7598 \\
ZH & PrA & SAN & 0.6488 & 0.6590 \\
ZH & PrG & SAN & 0.6094 & 0.6221 \\
ZH & Ours & SAN & \best{0.0437} & \best{0.0532} \\
\cmidrule(lr){2-5}
ZH & AbI & Full & 0.7559 & 0.7644 \\
ZH & Ours & Full & \best{0.0796} & \best{0.0905} \\
\bottomrule
\end{tabularx}
\PGTableNote{Bold marks the lowest exposure within the corresponding SAN or Full block. Presidio configurations are natively SAN-only and have no Full counterpart.}
\end{table}

\subsection{Screenshot and Reader-Audit Details}
\label{sec:appendix-screenshot}

This section retains the complete released-view comparisons and detailed screenshot diagnostics summarized in the main text. Tables~\ref{tab:appendix-crossturnpii-complete} and~\ref{tab:appendix-privscreen-complete} retain every SAN and Full row, while Table~\ref{tab:appendix-full-union-decomposition} separates per-view exposure from the final five-turn audited union. The remaining tables report type-stratified exposure, conditional re-exposure, method-conditioned reconstruction, OCR-backend sensitivity, and component cost.

\subsubsection{Complete Released-View Comparisons}

\begin{table}[!t]
\centering
\caption{Complete PrivScreen privacy exposure and visual-answer utility under the Original, SAN, and Full configurations.}
\label{tab:appendix-privscreen-complete}
\begin{tabularx}{\linewidth}{@{}L{1.1}C{0.8}C{1.025}C{1.025}C{1.025}C{1.025}@{}}
\toprule
& & \multicolumn{2}{c}{Qwen} & \multicolumn{2}{c}{Gemini} \\
\cmidrule(lr){3-4}\cmidrule(lr){5-6}
Method & View & Exposure$\downarrow$ & Utility$\uparrow$ & Exposure$\downarrow$ & Utility$\uparrow$ \\
\midrule
Original & Original & 0.9802 & 0.912 & 0.9811 & 0.968 \\
\midrule
AbI & SAN & 0.1992 & 0.914 & 0.2096 & 0.990 \\
PrA & SAN & 0.3683 & 0.896 & 0.3692 & 0.876 \\
PrG & SAN & 0.3579 & 0.844 & 0.3560 & 0.934 \\
Ours & SAN & \best{0.0585} & \best{0.916} & \best{0.0576} & \best{0.994} \\
\midrule
AbI & Full & 0.2134 & \best{0.914} & 0.2229 & 0.950 \\
Ours & Full & \best{0.0793} & 0.910 & \best{0.0831} & \best{0.966} \\
\bottomrule
\end{tabularx}
\PGTableNote{Bold marks the better endpoint within the corresponding SAN or Full comparison. Exposure and utility remain separate outcomes. Presidio configurations have no Full reference lifecycle.}
\end{table}

\begin{table*}[!t]
\centering
\caption{Per-view and five-turn decomposition of CrossTurnPII audited source-value leakage on common-gold sessions.}
\label{tab:appendix-full-union-decomposition}
\begin{tabularx}{\textwidth}{@{}C{0.4}L{0.8}R{0.6}C{1.3}C{1.3}C{1.0}C{1.3}C{1.3}C{1.0}@{}}
\toprule
& & & \multicolumn{3}{c}{AbI Full} & \multicolumn{3}{c}{Ours Full} \\
\cmidrule(lr){4-6}\cmidrule(lr){7-9}
Lang & Reader & Gold & Per-view & Five-turn & Increment & Per-view & Five-turn & Increment \\
\midrule
EN & Qwen & 1,394 & 0.3486 & 0.3486 & 0.0000 & 0.2066 & 0.2080 & 0.0014 \\
EN & Gemini & 1,394 & 0.3271 & 0.3271 & 0.0000 & 0.2145 & 0.2174 & 0.0029 \\
ZH & Qwen & 1,424 & 0.7022 & 0.7022 & 0.0000 & 0.1208 & 0.1222 & 0.0014 \\
ZH & Gemini & 1,424 & 0.7163 & 0.7163 & 0.0000 & 0.1405 & 0.1426 & 0.0021 \\
\bottomrule
\end{tabularx}
\PGTableNote{Per-view exposure is unioned within each session. The five-turn union additionally includes exact or partial reconstruction from the ordered Full upload views under the same matching rule. Increment is the five-turn union minus the per-view union.}
\end{table*}

\begin{table*}[!t]
\centering
\caption{PrivScreen field exposure by PII type. Each cell reports Qwen / Gemini; lower is better. SAN rows are placeholder-free erase-only views; Full rows are the two complete reference-releasing configurations.}
\label{tab:appendix-privscreen-per-type}
\begin{tabularx}{\textwidth}{@{}L{1.2}C{0.97}C{0.97}C{0.97}C{0.97}C{0.97}C{0.97}C{0.98}@{}}
\toprule
Method & Overall & Person & Email & Phone & Location & Account & IP address \\
\midrule
Original & .9802/.9811 & .9752/.9752 & .9867/.9867 & .9900/.9900 & .9700/.9800 & .9742/.9742 & 1/1 \\
AbI SAN & .1992/.2096 & .4134/.4257 & .0533/.0533 & 0/0 & .3400/.4000 & .0129/.0129 & 0/0 \\
PrA SAN & .3683/.3692 & .7574/.7500 & 0/0 & .1250/.1250 & .2300/.2600 & .2194/.2194 & .0400/.0600 \\
PrG SAN & .3579/.3560 & .1955/.1931 & .9200/.9133 & .0200/.0200 & .4600/.4600 & .7226/.7226 & 0/0 \\
Ours SAN & .0585/.0576 & .1139/.1114 & 0/0 & 0/0 & .1100/.1100 & 0/0 & .1000/.1000 \\
\midrule
AbI Full & .2134/.2229 & .4406/.4505 & .0533/.0800 & 0/0 & .3800/.4000 & .0129/.0129 & 0/0 \\
Ours Full & .0793/.0831 & .1634/.1733 & 0/0 & 0/0 & .1300/.1300 & 0/0 & .1000/.1000 \\
\bottomrule
\end{tabularx}
\PGTableNote{\textit{Note.} Person, email, and IP retain their benchmark fields. Phone groups mobile-number and phone-number field aliases; Location groups country/region and location; Account groups username, Xiaohongshu-ID, and Weixin-ID fields. Grouping changes only presentation, not the field scorer or exposure threshold.}
\end{table*}

\begin{table*}[!t]
\centering
\caption{Method-specific CrossTurnPII conditional re-exposure and signed net contrasts.}
\label{tab:eval-crossturnpii-reference-addition}
\begin{tabularx}{\textwidth}{@{}L{0.4}L{0.4}C{0.4}L{1.0}L{1.9}L{1.9}@{}}
\toprule
Reader & Method & Lang & New/protected & RER [95\% CI] & Net $\Delta$ [95\% CI] \\
\midrule
Qwen & AbI & EN & 240/2,018 & 0.1189 [0.0966,0.1442] & $+0.0784$ [$+0.0621,+0.0978$] \\
     &     & ZH & 20/706 & 0.0283 [0.0161,0.0411] & $+0.0046$ [$+0.0014,+0.0078$] \\
     & Ours & EN & 26/2,433 & 0.0107 [0.0050,0.0175] & $+0.0022$ [$-0.0053,+0.0104$] \\
     &      & ZH & 107/2,715 & 0.0394 [0.0274,0.0527] & $+0.0359$ [$+0.0240,+0.0490$] \\
\midrule
Gemini & AbI & EN & 199/2,011 & 0.0990 [0.0807,0.1177] & $+0.0645$ [$+0.0522,+0.0772$] \\
       &     & ZH & 16/682 & 0.0235 [0.0130,0.0334] & $+0.0046$ [$+0.0014,+0.0075$] \\
       & Ours & EN & 22/2,421 & 0.0091 [0.0044,0.0146] & $+0.0022$ [$-0.0028,+0.0073$] \\
       &      & ZH & 117/2,688 & 0.0435 [0.0300,0.0578] & $+0.0373$ [$+0.0243,+0.0507$] \\
\bottomrule
\end{tabularx}
\PGTableNote{\textit{Note.} RER is newly exposed Full entities divided by entities not recovered from the named method's placeholder-free SAN view. Net $\Delta$ is the signed Full-minus-SAN exposure difference and is retained as a secondary diagnostic. Intervals use 10,000 paired bootstrap resamples of all images belonging to the same persona (seed 1729).}
\end{table*}

\begin{table}[t]
\centering
\caption{Exact PrivacyGuard reference reconstruction by reader and language.}
\label{tab:appendix-targeted-reconstruction}
\begin{tabularx}{\linewidth}{@{}L{1.0}C{1.0}C{1.0}C{1.0}C{1.0}@{}}
\toprule
Reader & Single EN & Single ZH & Five-turn EN & Five-turn ZH \\
\midrule
Qwen & 23/2,510 & 43/2,550 & 42/1,454 & 43/1,393 \\
Gemini  & 25/2,510 & 55/2,521 & 39/1,454 & 38/1,365 \\
\bottomrule
\end{tabularx}
\PGTableNote{\textit{Note.} Cells are exact recoveries / emitted-reference targets. Small denominator differences preserve every target from each complete frozen reader run.}
\end{table}

The native $\mathrm{LR}_{\mathrm{detect}}$ columns provide the pre-render intermediate measurement for each default system, while the controlled-backend matrices serve as the same-OCR algorithm comparison and OCR-frontend ablation. Under native frontends, PrivacyGuard obtains 0.0724 versus 0.2018--0.4947 on PrivScreen and 0.0316/0.0389 versus 0.1733--0.4454/0.4509--0.6632 on controlled EN/ZH screenshots. Under shared PP-OCRv5-Mobile, the corresponding non-PrivacyGuard ranges are 0.1582--0.4750 and 0.0482--0.3397/0.2576--0.6001. PrivacyGuard retains GenR@1.0 between 0.81 and 0.86 on PrivScreen across the four backends. Tesseract creates heavier word-level fragmentation, where conservative inline bridging retains precision but lowers recall relative to the PP-OCR backends.

\begin{table*}[!t]
\centering
\caption{Controlled OCR-backend results on PrivScreen.}
\label{tab:appendix-screenshot-ocr-matrix}
\begin{tabularx}{\textwidth}{@{}L{0.8}L{1.2}C{1.0}C{1.0}C{1.0}C{1.0}C{1.0}@{}}
\toprule
Method & OCR & Det-P (core)$\uparrow$ & GenR@1.0$\uparrow$ & GenR@0.8$\uparrow$ & $\mathrm{LR}_{\mathrm{detect}}\downarrow$ & Mod$\downarrow$ \\
\midrule
\multirow{4}{*}{AbI}
 & EasyOCR & \best{0.7165} & 0.6410 & 0.6476 & 0.2018 & \best{2.1025} \\
 & PP-v5-Mobile & 0.6611 & 0.6743 & 0.6762 & \best{0.1582} & 2.5224 \\
 & PP-v5-Server & 0.7083 & \best{0.7029} & \best{0.7076} & 0.1636 & 2.5433 \\
 & Tesseract & 0.4141 & 0.3286 & 0.3581 & 0.3669 & 2.8747 \\
\cmidrule(lr){1-7}
\multirow{4}{*}{PrA}
 & EasyOCR & 0.3769 & 0.3848 & 0.3876 & 0.1882 & 6.3520 \\
 & PP-v5-Mobile & \best{0.4820} & \best{0.4200} & \best{0.4200} & \best{0.1779} & \best{4.6360} \\
 & PP-v5-Server & 0.3429 & 0.4114 & 0.4114 & 0.2036 & 6.0480 \\
 & Tesseract & 0.3654 & 0.3752 & 0.4048 & 0.2131 & 6.9700 \\
\cmidrule(lr){1-7}
\multirow{4}{*}{PrG}
 & EasyOCR & 0.4200 & 0.6476 & 0.6476 & 0.4968 & 6.0620 \\
 & PP-v5-Mobile & \best{0.4842} & \best{0.6705} & \best{0.6705} & 0.4750 & 6.2860 \\
 & PP-v5-Server & 0.4645 & 0.6676 & 0.6676 & \best{0.4647} & \best{5.7120} \\
 & Tesseract & 0.2182 & 0.4743 & 0.5333 & 0.4947 & 8.5480 \\
\cmidrule(lr){1-7}
\multirow{4}{*}{Ours}
 & EasyOCR & 0.8612 & 0.8257 & 0.8333 & 0.1514 & 2.3784 \\
 & PP-v5-Mobile & \best{0.9122} & \best{0.8600} & \best{0.8705} & \best{0.0724} & 2.3116 \\
 & PP-v5-Server & 0.9050 & 0.8133 & 0.8257 & 0.0978 & 2.2327 \\
 & Tesseract & 0.9084 & 0.8162 & 0.8219 & 0.1497 & \best{2.0534} \\
\bottomrule
\end{tabularx}
\PGTableNote{\textit{Note.} $\mathrm{LR}_{\mathrm{detect}}$ is the deterministic residual sensitive-character fraction after OCR and detection. Best values are marked within each method. The main text uses the shared PP-OCRv5-Mobile rows.}
\end{table*}

The frozen results archive retains the aggregate per-backend OCR and steady-state timing rows, the common-host component-cost summary, and the normal-start handset memory record. Table~\ref{tab:eval-runtime-cost} reports the paper's deployment-facing Android latency and memory measurements.

\subsection{Local Reference Resolution and Restoration Details}
\label{sec:appendix-restore-cases}

Table~\ref{tab:eval-restore-main} summarizes the 40 valid resolution-and-restoration cases that complete the local REF return path. Table~\ref{tab:appendix-restore-targeted} summarizes the implementation's rejection behavior on 59 targeted unsupported cases and 1,280 generated malformed or confusable strings. The anonymous source contains the deterministic case generator and exporter, and the frozen results archive retains the versioned summary JSON.

\begin{table*}[t]
\centering
\caption{Valid reference-resolution cases by coverage family.}
\label{tab:eval-restore-main}
\begin{tabularx}{\linewidth}{@{}L{1.3}C{0.2}L{1.5}@{}}
\toprule
Coverage family & Cases & Representative forms \\
\midrule
Canonical roots, projections, and length metadata
& 17 & \texttt{ADDR\#1}, \texttt{ADDR\#1.DIST}, \texttt{NAME\#4}, \texttt{ALNUM\#8.LEN=8} \\
Session-history reuse
& 9 & Turn-2 address, phone, number, and name roots/projections \\
Allowlisted normalized variants
& 14 & Whitespace, full-width, and separator-confusable forms \\
\bottomrule
\end{tabularx}
\PGTableNote{Every case resolves and restores the expected local value; brackets are omitted in the examples for compactness.}
\end{table*}

\begin{table*}[t]
\centering
\caption{Unsupported, malformed, and confusable reference cases by family.}
\label{tab:appendix-restore-targeted}
\begin{tabularx}{\linewidth}{@{}L{0.85}C{0.15}L{2}@{}}
\toprule
Invalid family & Cases & Representative references \\
\midrule
Unknown entity id & 3 & \texttt{[[ADDR\#999.CITY]]}, \texttt{[[PHONE\#999]]}, \texttt{[[NAME\#777]]}. \\
Foreign-session root absent locally & 3 & Address, email, and license-plate roots recorded only in another session and absent from the active-session mapping. \\
Future-turn unavailable & 3 & Address, ID-number, and alphanumeric roots that are not yet present in the current session state. \\
Type mismatch & 4 & \texttt{[[PHONE\#1]]}, \texttt{[[EMAIL\#4]]}, \texttt{[[NAME\#5]]}. \\
Unsupported scalar projection & 7 & \texttt{[[NAME\#4.CITY]]}, \texttt{[[ORG\#10.CITY]]}, \texttt{[[TIME\#11.ROAD]]}. \\
Unsupported address projection & 3 & \texttt{[[ADDR\#1.DOMAIN]]}, \texttt{[[ADDR\#1.LEN=4]]}, \texttt{[[ADDR\#1.TIME]]}. \\
Missing component projection & 3 & \texttt{[[ADDR\#1.PROV]]}, \texttt{[[ADDR\#1.COUNTRY]]}, \texttt{[[ADDR\#1.COUNTRY-PROV]]}. \\
Unknown label alias & 4 & \texttt{[[FOO\#1]]}, \texttt{[[ACCOUNT\#8.LEN=8]]}, and a full-width \texttt{CONTACT\#2} form. \\
Nested placeholder construction & 7 & \texttt{[[ADDR\#1.[CITY]]]}, \texttt{[[ADDR\#[1].CITY]]}, \texttt{[[ADDR\#1.CITY[[ROAD]]]]}. \\
Truncation or OCR mutation & 15 & \texttt{[[ADDR\#1}, \texttt{[[ADDR\#l.CITY]]}, \texttt{[[EMA1L\#5]]}, \texttt{[[LICENSE\_PLATE\#9]}. \\
Malformed or noncanonical form & 7 & \texttt{[[ADDR\#1..CITY]]}, \texttt{[[ADDR\#\#1.CITY]]}, \texttt{[[ADDR\#1.]]}, \texttt{[[addr\#1.CITY]]}. \\
\midrule
Generated malformed syntax & 1,024 & Invalid bracket, label, root, separator, and projection combinations. \\
Generated Unicode/confusable variants & 256 & Full-width and non-allowlisted label, digit, and separator variants. \\
\bottomrule
\end{tabularx}
\PGTableNote{\textit{Note.} All 59 targeted unsupported references and all 1,280 generated malformed or confusable strings remain unresolved in the evaluated implementation.}
\end{table*}

\section{Evaluation Protocol, Analysis, and Provenance}

This section consolidates deployment assumptions, metric definitions, analytical contracts, artifact provenance, and statistical procedures. The organization separates experimental results from the protocol and audit material used to interpret them.

\subsection{Deployment and Evaluation Topology}
\label{sec:appendix-topology}

PrivacyGuard's intended deployment is device-resident: mediation, raw values, and session mappings remain on the client. Comparative experiments place every protection pipeline on the same RTX 4090 server so all baselines can run under common hardware; this server is trusted experimental apparatus, not the remote planner or a deployment requirement. The handset profile instead runs PP-OCRv5-Mobile through ncnn and executes detection, binding, session state, rendering, resolution, and restoration inside the packaged Android host. A trusted controller invokes its loopback endpoint only to drive the workload and record measurements. The default in-memory session store retains records across turns, refreshes an idle timeout of 3,600 seconds in these experiments, and supports explicit deletion; persistence is opt-in for evaluation and audit.

\begin{table*}[t]
\centering
\caption{Execution profile associated with each evidence family.}
\label{tab:appendix-topology-map}
\begin{tabularx}{\linewidth}{@{}L{0.8}L{0.83}L{1.37}@{}}
\toprule
Evidence family & PrivacyGuard profile & Evidence supplied \\
\midrule
Text, address, screenshot metrics & Common RTX 4090 server & Algorithm and end-to-end output under shared inputs. \\
Controlled REF tests & Common RTX 4090 server & Binding and local-restoration mechanism behavior. \\
Audited planner-use paths & Server-hosted service + Android UI & Cloud sees mediated views; trusted base includes interface, runner, and server-hosted service. \\
Android runtime profile & Native OCR + embedded Python mediation in packaged Android host & Handset endpoint feasibility and component cost. \\
Live Android utility & Server-hosted protection + Android UI/harness & End-to-end planner utility under common protection-pipeline placement. \\
\bottomrule
\end{tabularx}
\end{table*}

\subsection{Local Privacy Repository and Session-State Separation}
\label{sec:appendix-store-separation}

The local privacy repository is a first-class, opt-in input to SAN detection rather than a restoration database. Its default location is \path{.privacyguard/privacy_repository.json}. A validated document contains persona records keyed by \path{persona_id}; each record holds typed slots for names, phones, payment and identity identifiers, email, structured addresses, organizations, and general alphanumeric values. Scalar slots may include match aliases, names retain structured full and component forms, and addresses retain typed component graphs. A write validates the complete patch, merges records by \path{persona_id}, atomically replaces the JSON file, and reloads the active rule-based detector. The repository file and its raw values are not part of the released cloud payload.

Table~\ref{tab:appendix-local-state-roles} distinguishes the three local state sources that otherwise risk being conflated. Static detection resources and the optional repository can affect which candidates SAN accepts, but neither can directly resolve a returned reference. Only records issued into the active session mapping store carry that authority.

\begin{table*}[htp]
\centering
\caption{Roles and authority of device-local state.}
\label{tab:appendix-local-state-roles}
\begin{tabularx}{\linewidth}{@{}L{0.8}L{0.7}L{1.5}@{}}
\toprule
State source & Lifetime & Permitted role \\
\midrule
Static detection resources & Packaged, read-only & Public priors, formats, contexts, and negative boundaries used only as candidate evidence. \\
Local privacy repository & Persistent, opt-in & User-linked typed values, aliases, and address/name components used only as detection evidence. \\
Session mapping store & Session/goal scoped; in-memory by default & Issued roots, typed records, projections, and raw values; the sole authority for local reference resolution and restoration. \\
\bottomrule
\end{tabularx}
\end{table*}

The evaluation isolates this mechanism. Every headline comparison starts with an empty repository. The matched diagnostic first records strict misses, enrolls only those values into the local repository, and reruns the same inputs under the same language and detection regime. Thus the repository row in Table~\ref{tab:eval-san-mechanisms-main} measures the incremental benefit of user-linked local evidence; it is not an additional cross-system baseline.

Let $P$ be the optional local privacy repository, $S_t$ the issued active-session records through turn $t$, and $a$ the cloud-returned text.
\begin{proposition}[Repository/session-state separation]
For fixed $S_t$ and $a$ and any repositories $P,P'$,
\begin{equation}
\Restore(a;S_t,P)=\Restore(a;S_t,P').
\label{eq:repository-session-separation}
\end{equation}
\end{proposition}
\begin{proof}
Resolution and restoration draw replacement records exclusively from $S_t$.
\end{proof}
The repository may affect earlier detection and therefore which records enter $S_t$; it cannot authorize an unissued value on the return path.

\subsection{Evaluation Metric Definitions}
\label{sec:appendix-metrics}

\paragraph{SAN leakage metrics.}
For gold occurrences $G$, let $|g|$ be the number of scored characters in $g$ and let $c(g)$ be the union of characters covered by safe type-compatible predictions. We define
\begin{equation}
\begin{aligned}
\operatorname{GenR}@\tau
&=\frac{|\{g\in G:c(g)/|g|\geq\tau\}|}{|G|},\\
\mathrm{LR}_{\mathrm{char}}
&=1-\frac{\sum_g c(g)}{\sum_g |g|}.
\end{aligned}
\end{equation}
Higher GenR indicates stronger occurrence coverage, and lower $\mathrm{LR}_{\mathrm{char}}$ indicates less direct residual leakage. The companion tables also report $\mathrm{EPred}_{gt}$, extra gold-hitting predictions per touched occurrence, and $\mathrm{ECov}_{char}$, unique over-coverage outside the useful gold-character union normalized by total gold characters. Det-P(core) is precision against core-tier PII labels and BBoxR@0.5 is region recall at IoU or gold-area-overlap threshold 0.5. Mod is modified regions per changed image, \path{modified_region_count}/\path{images_with_any_modification}; unchanged images are excluded. We define $\mathrm{LR}_{\mathrm{detect}}=1-R^{\mathrm{detect}}_{\mathrm{char}}$, where $R^{\mathrm{detect}}_{\mathrm{char}}$ is exact sensitive-character recall from the OCR-conditioned detector output. The native form retains each system's default OCR; the shared-OCR form fixes the frontend for component control. Both are pre-render diagnostics.

CrossTurnPII Full-view exposure scores exact or partial source-value extraction from the final image on a common gold-occurrence denominator. Because a Full protected view includes its reference overlay, this is an upload/reference-interface observation rather than a SAN direct-leakage measure. The same generic extraction probe on each method's placeholder-free, opaque SAN-only image supplies its rendered SAN endpoint. For AbI and PrivacyGuard separately, Equation~\ref{eq:method-matched-ref-increment} counts $n_{01}$, entities not recovered from SAN but recovered from Full, over $n_{00}+n_{01}$, all entities protected by that SAN view. The cross-method version first restricts to entities protected by both SAN views and compares Full recovery on that common set. The signed $\mathrm{Full}_m-\mathrm{SAN}_m=(n_{01}-n_{10})/N$ difference is retained as a secondary net diagnostic; Presidio has no REF contrast because it releases no REF surface. Let $n(\cdot)$ apply NFKC normalization, case folding, and removal of whitespace and the declared punctuation set, and let $\ell(a,b)$ be the longest common contiguous-substring length. For normalized gold $g$ and reader guess $h$, exact recovery holds when $g=h$ or when $|g|\geq2$ and $g$ is a substring of $h$. Excluding exact recovery, partial recovery holds when $|h|\geq2$ and $h$ is a substring of $g$, or when
\begin{equation}
\ell(g,h)\geq \max\{2,\min(\lfloor |g|/2\rfloor,6)\}.
\label{eq:partial-recovery-rule}
\end{equation}
Single-view reference reconstruction asks for the exact value behind a reference emitted by Full using one final image, while five-turn reference reconstruction may jointly use all ordered Full turns. Their exact-recovery rates over emitted-reference targets are mechanism-specific REF reconstruction-cost endpoints. The common-gold audit additionally scores the five-turn response by the exact/partial rule in Equation~\ref{eq:partial-recovery-rule}. PrivScreen retains its benchmark-native VQA metric: a gold field is exposed when its field-match score exceeds 0.6, and visual-answer accuracy is the fraction of 500 normal questions answered correctly. Its SAN rows provide the all-method RQ1 comparison, while its AbI and PrivacyGuard Full rows provide the RQ3 external final-view comparison on the same questions and gold fields. PrivScreen is not used to estimate the RQ2 REF increment.

\paragraph{Common-gold reader audit.}
Let $\mathcal S$ be the common evaluated sessions and $G_s$ the distinct normalized construction-derived private values in session $s$. For reader $A$, let $e_A(g,s)$ indicate exact or partial full-view exposure of value $g$ in any released turn, and let $r_A(g,s)$ indicate exact or partial recovery by the targeted session-level reference-reconstruction probe. We report
\begin{equation}
\begin{aligned}
\mathrm{ALR}(A)
&=\frac{\sum_{s\in\mathcal S}\sum_{g\in G_s}
\mathbf 1\{e_A(g,s)\lor r_A(g,s)\}}
{\sum_{s\in\mathcal S}|G_s|},\\
\mathrm{RefR}_{\mathrm{session}}(A)
&=\frac{\sum_{s,g}\mathbf 1\{r_A(g,s)\}}
{\sum_{s\in\mathcal S}|G_s|},\\
\Delta_{\mathrm{ref}}^{\mathrm{union}}(A)
&=\frac{\sum_{s,g}\mathbf 1\{\neg e_A(g,s)\land r_A(g,s)\}}
{\sum_{s\in\mathcal S}|G_s|}.
\end{aligned}
\label{eq:audited-leakage-rate}
\end{equation}
$\mathrm{RefR}_{\mathrm{session}}$ assigns every exact or partial targeted five-turn recovery from a Full pipeline to REF, including values also recovered by the per-view probe. $\mathrm{ALR}(A)$ is the Full-pipeline audited source-value leakage defined in Equation~\ref{eq:audited-leakage-rate}; $\Delta_{\mathrm{ref}}^{\mathrm{union}}$ is the targeted five-turn probe's net addition after deduplicating values already counted in per-view exposure. The method-specific and common-protected conditional re-exposure rates provide AbI's and PrivacyGuard's per-view REF costs. The targeted probes and upload/reference-interface ALR compare Full outputs from AbI and PrivacyGuard on the same sessions and normalized gold values under the two fixed readers.

\paragraph{REF and utility metrics.}

For labeled same-entity pairs at threshold $\tau$, $C_{\mathrm{pair}}$ is the fraction for which both mentions bind, $\mathrm{PCR}_{\mathrm{cond}}$ is consistency conditioned on both binding, and $\mathrm{PCR}_{\mathrm{strict}}=C_{\mathrm{pair}}\mathrm{PCR}_{\mathrm{cond}}$ treats an unbound mention as inconsistent. Unless explicitly marked conditional, PCR denotes this recall-aware strict form. Split-chain rate measures whether address fragments share the expected root across two- or three-block chains. REF also reports valid resolution-and-restoration accuracy for the evaluated plugin, with reference-grammar robustness reported separately. Android SR measures completed tasks. Audited planner-use paths are finite-suite counts.

The OCR-free/native source-coverage and residual metrics provide the main SAN evidence, with the placeholder-free SAN-only conditions supplying the all-method rendered-view comparison. CrossTurnPII conditional re-exposure and targeted single-view/five-turn probes over Full references provide REF reconstruction-cost measurements. PrivScreen Full outcomes remain an external RQ3 final-view comparison rather than an RQ2 increment. Per-view exposure and the common-gold five-turn upload-view union report the audited Full upload/reference interface without double-counting a value recovered by both reader probes.

\subsection{Binding and Leakage Analysis}
\label{sec:appendix-analytical}

This section states the binding and state-separation contracts used by the main text. The empirical studies remain the primary evidence for source removal, entity continuity, reconstruction, and task outcomes.

\subsubsection{Released Observation and Reference Cost}

Let $W$ be the hidden raw values in one session, $X$ the other released context, $D_b$ the raw source characters that remain under binding policy $b$, and $R_b$ the released reference graph. The upload-view observation is
\begin{equation}
O_b=(X,D_b,R_b).
\label{eq:unified-cloud-observation}
\end{equation}
Full adds type, root-equality, projection, layout, and optional-length cues to the corresponding SAN-only observation. A planner or decoder can ignore these cues, so their optimal utility and reconstruction envelopes cannot be lower than under SAN-only. The AndroidLab ablation measures realized planner utility, while the conditional re-exposure and targeted-reconstruction audits measure realized fixed-reader cost. These empirical endpoints, rather than the monotonic observation relation itself, determine the reported tradeoff.

\subsubsection{Binding Constraint and Dataset-Conditioned Risk}

For candidate $c$ and prior root $e$, let $\tau(\cdot)$ be entity type, $O(\cdot)$ the normalized parser record, $\operatorname{Conflict}$ the type-specific observed-conflict predicate, and $\operatorname{Share}$ the required positive-overlap predicate.

\begin{proposition}[Same-type observational non-separability]
\label{prop:binding-separation}
PrivacyGuard can reuse $e$ for $c$ only if
\begin{equation}
\tau(c)=\tau(e),\qquad
\neg\operatorname{Conflict}(O(c),O(e)),\qquad
\operatorname{Share}(O(c),O(e)).
\label{eq:compatible-reuse}
\end{equation}
Thus a gold-inconsistent reuse can occur only between same-type entities that remain non-separable under the binder's current observation. Scalar reuse additionally requires canonical equality.
\end{proposition}
\begin{proof}
Lookup uses type-specific indexes and returns only roots admitted by the corresponding compatibility predicate. Type mismatches and observed structured conflicts are rejected; scalar compatibility additionally checks canonical equality.
\end{proof}

Recent-first and fresh-on-ambiguity agree when exactly one compatible prior root exists. On an ambiguous gold continuation, fresh-on-ambiguity necessarily breaks prior-root continuity, whereas recent-first succeeds only when the most recent compatible root is the intended one. Exact canonical reuse also necessarily allocates a fresh root for cross-form continuations with no equal complete surface. These comparisons explain the continuity benefit of the mechanism but do not establish a general false-merge guarantee.

For a dataset $D$ in which each scored intended root is compatible, let $A$ denote multiple compatible prior roots and $F$ denote selection of a non-gold root. Then $F\subseteq A$, so
\begin{equation}
R_D=\Pr_D(F)=\Pr_D(A)\Pr_D(F\mid A)=q_D\rho_D.
\label{eq:dataset-semantic-impact}
\end{equation}
The controlled 480-session audit observes $0/4{,}352$ disagreements and $0/2{,}031$ impure multi-occurrence chains. The independent stress generator supplies complete compatible sets and intended-root labels: recent-first succeeds in 5,286/6,000 continuations, fresh-on-ambiguity in 4,236/6,000, and recent-first selects a non-gold compatible root in 714 trials. These values are conditioned on the declared generator. For a fixed policy-generated surface, changing only its gold agreement label does not change a fixed reader's input or reconstruction; each distinct policy surface therefore receives its own leakage audit.

\subsection{Benchmark Roles and Artifact Provenance}
\label{sec:appendix-provenance}

The evaluation combines public GUI privacy benchmarks with a controlled dataset designed for private reference mediation. Public benchmarks contribute external screenshot distributions, GUI trajectories, and region-level privacy labels. The controlled benchmark jointly supplies exact source values, bilingual text-to-image correspondence, cross-turn entity identities, and address projections, enabling SAN direct-leakage measurement, REF continuity and reconstruction scoring, and five-turn Full upload-view unions on common gold. PrivScreen and AndroidLab provide complementary external screenshot and task evidence, while the randomized multi-root generator provides an explicit policy-stress distribution. Table~\ref{tab:dataset-evidence-roles} records each evaluated source's annotations and evidentiary scope.

\begin{table*}[!t]
\centering
\caption{Evidence sources, available annotations, and evaluation scope.}
\label{tab:dataset-evidence-roles}
\begin{tabularx}{\linewidth}{@{}L{0.7}L{1}L{1.3}@{}}
\toprule
Evidence source & Available evidence & Role and scope \\
\midrule
CrossTurnPII controlled corpus
& Exact values and spans, rendered regions, text--image links, five-turn entity identity, and component projections
& Supports RQ1 SAN source removal and reader exposure, RQ2 attachment and reconstruction, and RQ3 controlled Full upload-view leakage. It does not provide external screenshot or broad task generality. \\

Structured-address stress
& 7,800 forms, component gold, 11 English and 12 Chinese perturbation families, and two-/three-block forms
& Supports RQ1 complete-source coverage and RQ2 cross-form continuity and PCR. It is a controlled form-distribution stress test rather than external workload evidence. \\

PrivScreen~\cite{zhang2025dualtap}
& 500 captured screens with visible-PII and geometry labels, benchmark privacy questions, and visual questions
& Supports external RQ1 SAN evaluation and RQ3 Full final-view privacy and visual utility. It does not provide cross-turn identity or local-resolution gold. \\

Fixed live-interface suite
& Fixed Android tasks, trajectory-issued references, resolved actions, and automatically scored answers
& Supports RQ2 verification of reference-required completion and planner-use paths. It is a targeted capability suite rather than a broad task benchmark. \\

AndroidLab~\cite{xu-etal-2025-androidlab}
& 138 public tasks over nine apps, fixed resets and judgments, task and action metrics, and retained trajectories
& Its SAN-only-to-Full contrast supports RQ2 mechanism attribution, while complete Full-system comparisons support RQ3 task utility. It does not provide exact source-extent, entity-identity, or projection gold. \\

Randomized multi-root stress
& Seeded compatible roots, generator-provided intended roots, ambiguity strata, and exact continuation outcomes
& Supports RQ2 binding-policy conformance and policy--intent agreement. Its outcomes are conditioned on the declared generator. \\

Xiaomi 12S workload
& Handset timing, memory, persistence, restoration, and completion records
& Supports RQ4 device-resident latency, memory, and completion measurements. It is not a comparative privacy or broad task benchmark. \\
\bottomrule
\end{tabularx}
\end{table*}

\paragraph{Controlled benchmark annotations.}
The controlled benchmark contains 2,400 rendered screenshots in 480 five-turn bilingual sessions. Structured persona records and deterministic scripts provide exact values, source regions, cross-turn identity, and address projections; scenario drafts were model-assisted and author-curated. No collected conversations or user records are used. The app-agnostic chat shell varies alignment, wrapping, density, free-flowing bubbles, field--value cards, and mixed layouts. Its 800-image reader subset contains 5,632 gold entities for single-image and five-turn reconstruction.

\paragraph{Structured-address stress construction.}
The address split re-instantiates 240 identities under 11 English and 12 Chinese form families, producing 3,720 English and 4,080 Chinese inputs with exact component and identity gold. It covers full and shortened forms, component views, abbreviations, reordered administration, OCR noise, and two-/three-block fragmentation. Because family coverage rather than real-world prevalence determines the sample, this is a mechanism stress test, not a population estimate; PrivScreen and AndroidLab supply complementary external evidence.

\paragraph{Baseline source and adapter audit.}
\label{sec:appendix-baseline-audit}
AbI~\cite{zhao2026anonymization} uses archived commit \texttt{5d70e9914acf9dda6f989cda6ac9402fffa47938}; its adapters invoke the pinned detection, exact value-map reuse, placeholder, and restoration paths. The evaluation-only \texttt{abi\_san\_only} condition keeps detected regions but releases opaque fill and no token metadata. Presidio Analyzer 2.2.362 uses commit \texttt{06616b33dc2b5d0cfbd9dc3eaf47f649024d84ad}; the GLiNER row uses \texttt{urchade/gliner\_multi\_pii-v1} at threshold 0.30. Both Presidio screenshot adapters are SAN-only. DualTAP~\cite{zhang2025dualtap} is not executable because its inspected entry point requires an unavailable generator checkpoint, and CAPED~\cite{shen2026caped} remains related work rather than a residual-content baseline. The anonymous source retains the pins, portability patches, adapters, and configurations.

\begin{table*}[!t]
\centering
\caption{Type-aware lexicon--gold content audit.}
\label{tab:appendix-lexicon-audit}
\begin{tabularx}{\textwidth}{@{}L{2}R{0.7}C{0.8}C{0.8}C{0.85}C{0.85}@{}}
\toprule
Corpus / type & Unique gold & Whole-value exact & Identifier exact & Structural atom median / p95 & Public-prior atom median / p95 \\
\midrule
Prompt EN / all     & 2,313 & 0 (0.00\%)  & 0 & 0.0000 / 0.5000 & 0.0000 / 0.6000 \\
\quad NAME          & 321   & 0 (0.00\%)  & 0 & 0.0000 / 0.0000 & 0.5556 / 0.6667 \\
\quad ADDRESS       & 598   & 0 (0.00\%)  & 0 & 0.2222 / 0.6000 & 0.0000 / 0.5833 \\
\addlinespace
Prompt ZH / all     & 2,186 & 15 (0.69\%) & 0 & 0.0000 / 0.2500 & 0.0000 / 0.3333 \\
\quad NAME          & 328   & 0 (0.00\%)  & 0 & 0.0000 / 0.0000 & 0.3333 / 0.3333 \\
\quad ADDRESS       & 474   & 15 (3.16\%) & 0 & 0.1429 / 0.4286 & 0.0769 / 0.5000 \\
\addlinespace
Address EN / ADDRESS & 880   & 15 (1.70\%) & 0 & 0.1282 / 0.2941 & 0.1429 / 0.4737 \\
Address ZH / ADDRESS & 1,039 & 0 (0.00\%)  & 0 & 0.0870 / 0.2500 & 0.0769 / 0.3333 \\
\bottomrule
\end{tabularx}
\PGTableNote{\textit{Note.} Exact equality uses NFKC normalization, case folding, and removal of punctuation and spacing. Atom columns report, for each unique gold value $g$ and type-applicable role $r$, the median and 95th percentile of $\max_{a\in r,\,a\subseteq g}|a|/|g|$. For ZH NAME, 0.3333 is the expected share of a one-character surname in a three-character full name; such surname and common-name atoms are public recognition priors and store no full-name association. NAME whole-value exact matches are zero. All 30 unique whole-value matches are standalone public \texttt{CITY} components; composite records and high-specificity identifiers have zero exact matches. The inventory classifier finds zero record-like positive or structural entries.}
\end{table*}

\begin{table*}[!t]
\centering
\caption{Conservative PrivacyGuard-only monotone sensitivity on common-gold sessions.}
\label{tab:appendix-alr-monotone-sensitivity}
\begin{tabularx}{\textwidth}{@{}L{0.7}C{0.45}C{1.0}L{1.4}L{1.45}@{}}
\toprule
Reader & Lang & AbI Full ALR & Ours SAN-only-or-Full ALR & $\Delta$ (95\% CI) \\
\midrule
Qwen   & EN & 0.3486 & 0.2202 & $-0.1284$ [$-0.1607,-0.0937$] \\
       & ZH & 0.7022 & 0.1264 & $-0.5758$ [$-0.6094,-0.5408$] \\
Gemini & EN & 0.3271 & 0.2260 & $-0.1012$ [$-0.1344,-0.0670$] \\
       & ZH & 0.7163 & 0.1503 & $-0.5660$ [$-0.5970,-0.5353$] \\
\bottomrule
\end{tabularx}
\PGTableNote{\textit{Note.} This sensitivity charges PrivacyGuard for a gold value exposed in either the SAN-only ablation or the Full per-view output, then unions the exact/partial five-turn reconstruction. It complements the matched Full-pipeline comparison in Table~\ref{tab:eval-full-leakage}. Intervals use the same 10,000 persona-paired resamples.}
\end{table*}

\begin{table*}[!t]
\centering
\caption{Bootstrap intervals for the main text/address SAN comparison.}
\label{tab:appendix-bootstrap-summary}
\begin{tabularx}{\textwidth}{@{}L{0.7}C{0.4}L{1.1}L{1.4}L{1.4}@{}}
\toprule
Setting & Lang & Method & GenR@1.0 [95\% CI] & $\mathrm{LR}_{\mathrm{char}}$ [95\% CI] \\
\midrule
Text & EN & PrG & 0.7912 [0.7808, 0.8015] & 0.1283 [0.1202, 0.1359] \\
Text & EN & Ours & 0.9740 [0.9619, 0.9846] & 0.0125 [0.0075, 0.0186] \\
Text & ZH & PrA & 0.5048 [0.4710, 0.5389] & 0.3998 [0.3808, 0.4193] \\
Text & ZH & Ours & 0.9971 [0.9931, 0.9997] & 0.0023 [0.0002, 0.0054] \\
Address & EN & PrG & 0.6247 [0.6127, 0.6352] & 0.0537 [0.0465, 0.0612] \\
Address & EN & Ours & 0.8981 [0.8691, 0.9238] & 0.0430 [0.0318, 0.0559] \\
Address & ZH & PrG & 0.6567 [0.6481, 0.6647] & 0.0114 [0.0077, 0.0156] \\
Address & ZH & Ours & 0.9572 [0.9447, 0.9683] & 0.0071 [0.0051, 0.0094] \\
\bottomrule
\end{tabularx}
\PGTableNote{\textit{Note.} Intervals use 5,000 group resamples with seed 1729 and preserve within-persona correlation. Point estimates match Table~\ref{tab:eval-text-address-main}.}
\end{table*}

\paragraph{Repository and lexicon snapshot audit.}
All headline experiments start with an empty local privacy repository. Detector inputs consist of static lexicon resources, deterministic parsing rules, and session evidence accumulated online during the current task. A content audit compares all 33 JSON lexicon files with every unique prompt/address gold value, records any whole-value overlap, and verifies that annotation labels, per-example private records, full NAME values, and high-specificity identifiers are absent from detector inputs. The controlled screenshots are rendered from the prompt records, so their source gold is counted once. The snapshot contains 7,842 file-scoped entries classified as public-prior atoms, format/structure patterns, contextual boundaries, or negative/boundary terms. The inventory classifier finds zero record-like positive or structural entries; its single record-like heuristic flag is a negative-name blocklist item.

The revised audit restricts each atom comparison to the PII families for which that resource is consumed---for example, a surname-list character is relevant to a \texttt{NAME}, not merely to the same character occurring inside an address. Table~\ref{tab:appendix-lexicon-audit} reports both complete-value equality and the share contributed by the largest single applicable atom.

Matched atoms consist of structural keys such as city/district/road/residential-compound markers and street or unit suffixes, common-name atoms used by name rules, public administrative or geographic names, and format prefixes. These matches quantify the public parsing support contributed by individual atoms. Complete-value equality and largest-atom coverage in Table~\ref{tab:appendix-lexicon-audit} provide the record-separation audit. The anonymous source includes the type-applicability map and audit implementation; Table~\ref{tab:appendix-lexicon-audit} reports the per-type summary used here, and Appendix~\ref{sec:appendix-prompt} reports the post-ingest sensitivity diagnostic.

\subsection{Statistical Results}
\label{sec:appendix-statistics}

\paragraph{Statistical protocol.}
For the prompt and structured-address headline metrics, we report percentile 95\% bootstrap confidence intervals. Resampling is performed over whole personas with replacement---\texttt{persona\_id} for prompts and \texttt{persona\_consistency\_key} for address variants---so correlated turns or variants from one private profile remain in the same resampled group. We use 5,000 resamples and PRNG seed 1729. Table~\ref{tab:appendix-bootstrap-summary} reports intervals for the main SAN text/address comparison. The CrossTurnPII REF audit pairs each method's Full and SAN-only outcomes and resamples the 40 complete personas (four five-turn sessions per persona) 10,000 times with seed 1729. For Qwen/Gemini, the pooled method-specific conditional re-exposure rates are 0.0955 [0.0741,0.1171]/0.0798 [0.0625,0.0960] for AbI and 0.0258 [0.0177,0.0351]/0.0272 [0.0182,0.0372] for PrivacyGuard. Table~\ref{tab:eval-crossturnpii-reference-addition} keeps language strata and secondary signed contrasts, while Table~\ref{tab:eval-ref-reconstruction-cost} reports the paired common-protected comparison. The common-gold audited-leakage comparison uses the same persona-paired protocol within each language and recomputes both Full pipelines on every resample. Table~\ref{tab:eval-full-leakage} reports all four PrivacyGuard-minus-AbI intervals, and Table~\ref{tab:appendix-full-union-decomposition} separates the per-view and five-turn unions. PrivScreen reports absolute SAN and Full outcomes from the same 500 images and fixed readers but supplies no RQ2 increment. Controlled mechanism, restoration, and automatically verified live-interface results are finite-suite counts. Android comparisons use task-paired bootstrap intervals over the same 138 tasks (Table~\ref{tab:eval-android-paired}). A separate 5,000-resample whole-\texttt{persona\_consistency\_key} bootstrap provides the strict-PCR intervals for Table~\ref{tab:eval-ref-capability-continuity}; the retained bootstrap JSON reports pair-coverage, conditional-PCR, and strict-PCR intervals for the four prespecified audited rows.

\end{document}